\documentclass[sigconf,nonacm]{acmart}

\AtBeginDocument{%
  }

\setcopyright{acmcopyright}
\copyrightyear{2018}
\acmYear{2018}
\acmDOI{XXXXXXX.XXXXXXX}

\acmConference[Conference acronym 'XX]{Make sure to enter the correct
  conference title from your rights confirmation emai}{June 03--05,
  2018}{Woodstock, NY}
\acmPrice{15.00}
\acmISBN{978-1-4503-XXXX-X/18/06}

\usepackage{balance}

\usepackage{bm}
\usepackage[linesnumbered,ruled,vlined]{algorithm2e}

\newcommand{\argmin}{\mathop{\rm argmin}}
\newcommand{\obj}{f_{(G,v)}}

\DeclareMathOperator{\proj}{proj}

\theoremstyle{plain}
\newtheorem{theorem}     {Theorem}
\newtheorem{lemma}       {Lemma}

\newtheorem{corollary}   {Corollary}

\newtheorem{problem}     {Problem}

\newtheorem{fact}        {Fact}
\newtheorem{assumption}  {Assumption}

\newcommand{\squishlist}{
 \begin{list}{$\bullet$}
  {  \setlength{\itemsep}{0pt}
     \setlength{\parsep}{3pt}
     \setlength{\topsep}{3pt}
     \setlength{\partopsep}{0pt}
     \setlength{\leftmargin}{2em}
     \setlength{\labelwidth}{1.5em}
     \setlength{\labelsep}{0.5em}
} }

\newcommand{\squishend}{
  \end{list}
}

\begin{document}

\title{Local Centrality Minimization with Quality Guarantees}

\author{Atsushi Miyauchi}
\email{atsushi.miyauchi@centai.eu}
\affiliation{%
  \institution{CENTAI Institute}
  \city{Turin}
  \country{Italy}
}
\author{Lorenzo Severini}
\email{lorenzo.severini@unicredit.eu}
\affiliation{%
  \institution{UniCredit}
  \city{Rome}
  \country{Italy}
}

\author{Francesco Bonchi}
\email{francesco.bonchi@centai.eu}
\affiliation{%
  \institution{CENTAI Institute, Turin, Italy}
  \institution{Eurecat, Barcelona, Spain}
  \city{\empty}
  \country{\empty}
}


\begin{abstract}
Centrality measures, quantifying the importance of vertices or edges, play a fundamental role in network analysis.
To date, triggered by some positive approximability results,
a large body of work has been devoted to studying centrality maximization,
where the goal is to maximize the centrality score of a target vertex by manipulating the structure of a given network.
On the other hand, due to the lack of such results,
only very little attention has been paid to centrality minimization, despite its practical usefulness.

In this study, we introduce a novel optimization model for local centrality minimization, where the manipulation is allowed only around the target vertex.
We prove the NP-hardness of our model and that the most intuitive greedy algorithm has a quite limited performance in terms of approximation ratio.
Then we design two effective approximation algorithms:
The first algorithm is a highly-scalable algorithm that has an approximation ratio unachievable by the greedy algorithm,
while the second algorithm is a bicriteria approximation algorithm that solves a continuous relaxation based on the Lov\'asz extension, using a projected subgradient method.
To the best of our knowledge, ours are the first polynomial-time algorithms with provable approximation guarantees for centrality minimization.
Experiments using a variety of real-world networks demonstrate the effectiveness of our proposed algorithms: Our first algorithm is applicable to million-scale graphs
and obtains much better solutions than those of scalable baselines,
while our second algorithm is rather strong against adversarial instances.
\end{abstract}

%
%






\maketitle

\section{Introduction}\label{sec:intro}
Among the many analytical tools that social network analysis~\cite{Knoke+20} borrowed from graph theory, centrality measures play a fundamental role in a wide variety of analyses~\cite{Das+18}.
Centrality, which quantifies the \emph{importance} of vertices or edges only based on the graph structure, has found many applications, including, e.g., identification of important users or connections in social networks,
community detection~\cite{Newman+04}, anomaly detection~\cite{Maulana+21}, to name a few.

\emph{Local centrality minimization} is the problem of removing a few existing edges around a target vertex, so as to minimize its centrality score. 
A direct application can be found in the context of reducing the visibility, or influence, of a targeted harmful user in a social network, without explicitly blocking the user account.
In this scenario, to minimize the impact to the other decent users, it is quite reasonable to manipulate only the ties around the targeted user. 
The edges to be removed, identified by local centrality minimization, 
are the most important edges for the centrality (i.e., visibility or influence) of the target vertex. 
In this regard, another direct application is to keep satisfying influential users so that they are engaged in the platform.
The degree of satisfaction of such influencers often depends on how actually influential they are in the platform, i.e., how much their content is consumed by the network.
Local centrality minimization can be used for revealing the most important connections between the influencers and their followers, which are key in contributing to their  visibility.

While a large body of work has been devoted to studying \emph{centrality maximization}~\cite{Das+18}, much less attention has been paid to centrality minimization
(see Section \ref{sec:related} for a brief literature survey).
Generally speaking, the goal of centrality maximization is to maximize the centrality score of a target vertex by adding a limited number of edges to the network.
In many reasonable optimization models, the objective function becomes monotone and submodular~\cite{Fujishige05},
and thus a simple greedy algorithm admits a $(1-1/\mathrm{e})$-approximation~\cite{Nemhauser+78}.
This positive result makes the basis of various studies on centrality maximization~\cite{Bergamini+18,Castaldo+20,Crescenzi+16,DAngelo+19,Ishakian+12,Medya+18}.

The lack of positive approximability results has instead limited the attention on the centrality minimization problem, especially in its local variant.
Waniek et al.~\cite{Waniek+18} introduced several optimization models for local centrality minimization under some specific objectives and constraints, investigated the computational complexity of their models, and devised some algorithms.
Later, Waniek et al.~\cite{Waniek+23} investigated local centrality minimization from a game-theoretic point of view.
However, the work by Waniek et al.~\cite{Waniek+18,Waniek+23} proposed only (exponential-time) exact algorithms and heuristics.

In this paper, we study the local centrality minimization problem, adopting
the most well-established centrality measure called the \emph{harmonic centrality}~\cite{Das+18},
which quantifies the importance of vertices based on the level of reachability from the other vertices.
The harmonic centrality is known as an effective alternative to the closeness centrality~\cite{Boldi_Vigna_14},
which was employed in Waniek et al.~\cite{Waniek+18},
in the sense that unlike the closeness centrality, it is well-defined even in the case where a graph is not strongly connected.

Boldi and Vigna~\cite{Boldi_Vigna_14} showed that among all the known centrality measures,
only the harmonic centrality satisfies all the desirable axioms, namely the size axiom, density axiom, and score monotonicity axiom.
Recently, Murai and Yoshida~\cite{Murai_Yoshida_19} theoretically and empirically demonstrated that among well-known centrality measures,
the harmonic centrality is most stable (thus reliable) against the uncertainty of a given graph.

\subsection{Paper contributions and roadmap}
In this paper, we introduce a novel optimization model for local centrality minimization, where the harmonic centrality is employed as an objective function.
Specifically, in our model, given a directed graph $G=(V,A)$, a target vertex $v\in V$, and a budget $b\in \mathbb{Z}_{>0}$,
we aim to find a set of incoming edges of $v$ with size (no greater than) $b$ whose removal minimizes the harmonic centrality score of $v$, denoted by $h_G(v)$ (to be defined in Section~\ref{sec:problem}).

For our optimization model, we first analyze the computational complexity.
Specifically, we show that our model is NP-hard even on a very limited graph class (i.e., acyclic graphs),
by constructing a polynomial-time reduction from the minimum $k$-union problem.
Furthermore, we prove that the most intuitive greedy algorithm,
which iteratively removes an incoming edge of $v$ that maximally decreases the objective value, cannot achieve an approximation ratio of $o(|V|)$,
while any reasonable algorithm has an approximation ratio of $O(|V|)$.
This negative result motivates the design of algorithms that exploit the characteristics of our model.

We design two polynomial-time approximation algorithms.
The first algorithm is a highly-scalable algorithm that has an approximation ratio of $\sqrt{2h_G(v)}$.
We stress that as $\sqrt{2h_G(v)}= O(\sqrt{|V|})=o(|V|)$, this approximation ratio is unachievable by the above greedy algorithm.
Our algorithm first sorts the incoming neighbors of the target vertex $v$ in the decreasing order of their harmonic centrality scores on a slightly modified graph,
and then removes $b$ incoming edges from the top-$b$ vertices in the sorted list.
To prove the approximation ratio, we scrutinize the relationship between the harmonic centrality scores of the target vertex and its incoming neighbors.
In the end, we also prove the tightness of our analysis of the approximation ratio.

The second algorithm is a polynomial-time algorithm
that has a bicriteria approximation ratio of $(\frac{1}{\alpha},(\frac{1}{1-\alpha},\epsilon))$ for any $\alpha \in (0,1)$ and $\epsilon >0$.
That is, the algorithm finds a subset of incoming edges of the target vertex $v$ with size at most $b/\alpha$
but attains the objective value at most the original optimal value times $\frac{1}{1-\alpha}$ plus $\epsilon$.
Therefore, the algorithm approximates the original optimal value while violating the budget constraint to some bounded extent. 
To design the algorithm, we first introduce a continuous relaxation of our model.
To this end, we use the well-known extension of set functions, called the Lov\'asz extension~\cite{Lovasz83}.
An important fact is that the objective function of our model is submodular,
which guarantees that its Lov\'asz extension is (not necessarily differentiable but) convex.
Therefore, we can solve the relaxation (with an arbitrarily small error) using a projected subgradient method~\cite{Beck17}.
Once we get a fractional solution, we apply a simple probabilistic procedure and obtain a subset of incoming edges of the target vertex $v$.

Finally, our experiments on a variety of real-world networks show that our first algorithm is applicable to million-scale graphs
and obtains much better solutions than those of scalable baselines,
while our second algorithm is strong against adversarial instances.

In summary, our contributions are as follows:
\begin{itemize}
\leftskip=-5pt
\item We study the \emph{local harmonic centrality minimization problem}: We prove that it is NP-hard even on acyclic graphs, its objective function is submodular, and the most intuitive greedy algorithm cannot achieve $o(|V|)$-approximation (Section~\ref{sec:problem}).
\item We devise a highly-scalable algorithm with an approximation ratio of $\sqrt{2h_G(v)}$, which is unachievable by the greedy algorithm (Section~\ref{sec:fast}). 
\item We then devise a bicriteria approximation algorithm that solves a continuous relaxation based on the Lov\'asz extension, using a projected subgradient method (Sections~\ref{sec:bicriteria}~and~\ref{sec:algorithm_cont}).
\end{itemize}
To the best of our knowledge, ours are the first polynomial-time algorithms with provable approximation guarantees for centrality minimization.

\section{Related Work}\label{sec:related}
In this section, we review related literature about centrality minimization and maximization, submodular minimization, and other relevant applications in social networks analysis.

\smallskip
\noindent \textbf{Centrality minimization.}
The most related to the present paper is the work on centrality minimization by Waniek et al.~\cite{Waniek+18,Waniek+23} whose goal is to provide a methodology that contributes to hiding individuals in social networks from centrality-based network analysis algorithms.

More specifically, Waniek et al.~\cite{Waniek+18} introduced the following optimization model:
Given a directed graph $G=(V,A)$, a target vertex $v\in V$, a budget $b\in \mathbb{Z}_{>0}$, and a set $M$ of possible single-edge modifications,
we are asked to select at most $b$ actions in $M$
so as to minimize the centrality score of $v$ while satisfying some constraint on the influence of some vertices in the graph.
As a centrality measure, they considered the degree centrality, closeness centrality, and betweenness centrality.
They showed that except for the degree centrality case, the above model is NP-hard even when the constraint on the influence of vertices is ignored, and devised simple heuristics.

Waniek et al.~\cite{Waniek+23} investigated local centrality minimization from a game-theoretic point of view.
As a tool to analyze their game, they studied the following optimization model:
Given a directed graph $G=(V,A)$, a target vertex $v\in V$, a hiding parameter $\delta$, and a set $M$ as above, 
we are asked to find a minimal subset of $M$ guaranteeing that there are at least $\delta$ vertices having a centrality score greater than that of $v$.
As a centrality measure, they again considered the above three.
They showed that the model is 2-approximable for the degree centrality but is inapproximable within any logarithmic factor for the other two.
Note that the above approximation is just for the size of the output rather than the ranking of the centrality score of $v$ (or the centrality score of $v$).

Veremyev et al.~\cite{Veremyev+19} studied a global centrality minimization problem:
Given an undirected graph $G=(V,E)$ with cost $c_e\in \mathbb{R}_{\geq 0}$ for each $e\in E$, a target vertex subset $S\subseteq V$, and a budget $b\in \mathbb{Z}_{>0}$,
we are asked to find $F\subseteq E$ whose removal minimizes the centrality score of the target vertex subset $S$ subject to the budget constraint $\sum_{e\in F}c_e\leq b$.
The centrality score of a vertex subset is defined as a generalization of the centrality score of a vertex.
As a centrality measure, they considered a quite general one, based on distance between vertices,
which includes the harmonic centrality as a special case.
They proved that the above model is NP-hard for any centrality measure included in the above,
and as a by-product of the analysis, they also mentioned the NP-hardness of its local variant, which coincides with (the undirected-graph counterpart of) our proposed model.
In the present paper, by focusing on the harmonic centrality, we prove that our model is NP-hard even on a very limited graph class (i.e., acyclic graphs).
On a positive side, they presented an exact algorithm based on mathematical programming and greedy heuristics.
Very recently, Liu et al.~\cite{Liu+23} addressed another global centrality minimization problem,
where the objective function is a centrality measure called the information centrality and the connectivity of the resulting graph is guaranteed.

\smallskip
\noindent \textbf{Centrality maximization.}
Centrality maximization has more actively been studied in the literature
(e.g., \cite{Bergamini+18,Castaldo+20,Crescenzi+16,DAngelo+19,Ishakian+12,Medya+18}),
where the most related to ours is due to Crescenzi et al.~\cite{Crescenzi+16}.
They introduced the harmonic centrality maximization problem,
where given a directed graph $G=(V,A)$, a target vertex $v\in V$, and a budget $b\in \mathbb{Z}_{>0}$,
we are asked to insert at most $b$ incoming edges of $v$ so as to maximize the harmonic centrality score of $v$.
Our proposed optimization model can be seen as a minimization counterpart of their problem.
They proved that the problem is APX-hard, but 
devised a polynomial-time $(1-1/\mathrm{e})$-approximation algorithm based on the submodularity of the objective function.
Finally, we note that there is another class of problems also called centrality maximization,
where the goal is to find $S\subseteq V$ that has the maximum group centrality score (e.g.,~\cite{Angriman+20,Angriman+21,Bergamini+18,Chen+16,Li+19,Ishakian+12,Mahmoody+16,Mumtaz+17,Pellegrina23,Zhao+17}), which is less relevant to the present paper.

\smallskip
\noindent \textbf{Submodular minimization.}
Submodular minimization is one of the most well-studied problem classes in combinatorial optimization.
Among the literature, the most related work is due to Svitkina and Fleischer~\cite{Svitkina_Fleischer_11}.
They stated that a polynomial-time $(\frac{1}{\alpha},\frac{1}{1-\alpha})$-bicriteria approximation algorithm
for submodular minimization with a cardinality upper bound (and thus for our proposed model) is possible for any $\alpha \in (0,1)$,
using techniques in Hayrapetyan et al.~\cite{Hayrapetyan+_05}.
However, Hayrapetyan et al.~\cite{Hayrapetyan+_05} addressed another problem called the minimum-size bounded-capacity cut, 
where the function in the constraint instead of the objective function is submodular, which is not the case in submodular minimization with a cardinality upper bound.
Therefore, the above statement is not trivial and even our proposed model should be handled in a formal way.
Lov\'asz extension has actively been used for developing novel network analysis algorithms (e.g., \cite{Konar+21,Konar+22,Rangapuram+13}).

\smallskip
\noindent \textbf{Applications.}
Reducing the visibility or influence of target users in social networks has been studied in the context of influence minimization~\cite{Khalil+13,Wang+20,Yang+19}. 
All existing studies are based on some influence diffusion models such as the independent cascade model~\cite{Goldenberg+01_1,Goldenberg+01_2} and the linear threshold model~\cite{Kempe+03}.
Unlike those, our model does not assume any influence diffusion model, but is just based on the network structure.
Very recently, Fabbri et al.~\cite{Fabbri+22} and Coupette et al.~\cite{Coupette+23} addressed the problem of reducing the exposure to harmful contents in social media networks.  

On the other hand, identifying the users and/or connections that play a key role for user engagement in social networks has also attracted much attention.
Bhawalkar et al.~\cite{Bhawalkar+15} initiated this kind of study from an optimization perspective.
They invented a model that aims to find a group of users whose permanent use of the service guarantees user engagement as much as possible,
and designed polynomial-time algorithms for some cases.
Later, Zhang et al.~\cite{Zhang+17} and Zhu et al.~\cite{Zhu+18} introduced variants of the above model, and devised intuitive heuristics.

\section{Problem Formulation and Characterization}\label{sec:problem}
In this section, we mathematically formulate our problem (Problem 1), 
and prove its NP-hardness (Theorem 1) and the submodularity of the objective function (Theorem 2). Finally, we show the quite limited performance of the greedy algorithm (Theorem 3). 

Let $G=(V,A)$ be a directed graph (or \emph{digraph} for short). 
Throughout the paper, we assume that digraphs are simple, that is, there exist neither self-loops nor multiple edges. 
For $F\subseteq A$, we define $G\setminus F$ as the subgraph of $G$ that is constructed by removing all edges in $F$ from $G$, i.e., $G\setminus F= (V,A\setminus F)$.
For $v\in V$, we denote by $\rho(v)$ the set of incoming edges of $v$, i.e., $\rho(v)=\{(u,v)\in A\mid u\in V\}$.
For $v\in V$, let $h_G(v)$ be the harmonic centrality score of $v$ on a digraph $G$, i.e.,
\begin{align*}
h_G(v)=\sum_{u\in V\setminus \{v\}}\frac{1}{d_G(u,v)},
\end{align*}
where $d_G(u,v)$ is the (shortest-path) distance from $u\in V$ to $v\in V$ on $G$, and by convention, $d_G(u,v) = \infty$ when $v$ is not reachable from $u$. Note that, contrarily to other centrality measures, even in the case where a digraph is not strongly connected, the harmonic centrality is still well-defined (assuming by convention that $1/\infty = 0$).
Intuitively, the harmonic centrality quantifies the importance of a given vertex $v$ based on the level of reachability from the other vertices.
The problem we tackle in this paper is formalized as follows:
\begin{problem}[Local harmonic centrality minimization]\label{prob:centrality_min}
Given a digraph $G=(V,A)$, a target vertex $v\in V$, and a budget $b\in \mathbb{Z}_{>0}$,
we are asked to find $F\subseteq \rho(v)$ with $|F|\leq  b$ whose removal
minimizes the harmonic centrality of $v\in V$, i.e., $\obj(F)\coloneqq h_{G\setminus F}(v)$.
\end{problem}

By constructing a polynomial-time reduction from the NP-hard optimization problem called the \emph{minimum $k$-union}, we can prove the following. 
The proof can be found in Appendix~\ref{appendix:proof_hardness}. 
\begin{theorem}\label{thm:hardness}
Problem~\ref{prob:centrality_min} is NP-hard even on acyclic graphs.
\end{theorem}

We next show that the objective function $\obj$ of Problem~\ref{prob:centrality_min} is submodular,
which helps us design our bicriteria approximation algorithm in Section \ref{sec:bicriteria}.
Let $S$ be a finite set.
A set function $f\colon 2^S\rightarrow\mathbb{R}$ is said to be \emph{submodular} if for any $X,Y\subseteq S$, it holds that
\begin{align*}
f(X)+f(Y)\geq f(X\cup Y) + f(X\cap Y).
\end{align*}
We prove the following in Appendix~\ref{appendix:proof_submodular}: 
\begin{theorem}\label{thm:submodular}
For any $G=(V,A)$ and $v\in V$, the objective function $\obj$ of Problem 1 is submodular.
\end{theorem}

Finally, we prove that the most intuitive greedy algorithm does not have any non-trivial approximation ratio. 
Specifically, we consider the algorithm that iteratively removes an incoming edge of the target vertex $v$
that maximally decreases the harmonic centrality score of $v$, until it exhausts the budget.
For reference, the pseudo-code is given in Algorithm~\ref{alg:greedy} in Appendix~\ref{appendix:greedy}.
This algorithm runs in $O(b|\rho(v)|(|V|+|A|))$ time.
Note that, unlike many submodular maximization algorithms,
the \emph{lazy evaluation} technique~\cite{Minoux05} cannot be used to obtain a practically efficient implementation. 
The proof of the following is available in Appendix~\ref{appendix:proof_greedy}. 

\begin{theorem}\label{prop:greedy}
The greedy algorithm has no approximation ratio of $o(|V|)$ for Problem~\ref{prob:centrality_min},
while any algorithm that outputs $F\subseteq \rho(v)$ with $|F|=b$ has an approximation ratio of $O(|V|)$.
\end{theorem}

\section{Scalable Approximation Algorithm}\label{sec:fast}
In this section, we present a highly-scalable $\sqrt{2h_G(v)}$-approximation algorithm for Problem~\ref{prob:centrality_min}. 

\subsection{Algorithm}
Let $N_\text{in}(v)$ be the set of incoming neighbors of $v$, i.e., $N_\text{in}(v)=\{w\in V\mid (w,v)\in \rho(v)\}$. 
The intuition behind our algorithm is quite simple: 
As long as there exists a vertex $w\in N_\text{in}(v)$ that has a large harmonic centrality score, so does the target vertex $v$. 
This means that it is urgent to remove incoming edges of $v$ that come from vertices having large harmonic centrality scores. 
Note that our algorithm and analysis consider the harmonic centrality scores on $G\setminus \rho(v)$ rather than $G$; 
this is essential to obtain our approximation ratio. 
Specifically, our algorithm first sorts the elements of $N_\text{in}(v)$ as $(w_1,\dots, w_{|\rho(v)|})$ so that 
$h_{G\setminus \rho(v)}(w_1)\geq \cdots \geq h_{G\setminus \rho(v)}(w_{|\rho(v)|})$ 
and just returns $\{(w_1,v),\dots, (w_b,v)\}$. 
For reference, the entire procedure is described in Algorithm~\ref{alg:fast}. 

\begin{algorithm}[t]
\caption{$\sqrt{2h_G(v)}$-approximation algorithm for Problem~\ref{prob:centrality_min}}
\label{alg:fast}
\SetKwInOut{Input}{Input}
\SetKwInOut{Output}{Output} \Input{\ $G=(V,A)$, $v\in V$, and $b\in \mathbb{Z}_{>0}$}
\Output{\ $F\subseteq \rho(v)$ with $|F|\leq b$}
Sort the elements of $N_\mathrm{in}(v)$  as $(w_1,\dots, w_{|\rho(v)|})$
so that $h_{G\setminus \rho(v)}(w_1)\geq \cdots \geq h_{G\setminus \rho(v)}(w_{|\rho(v)|})$\;
\Return $\{(w_1,v),\dots, (w_b,v)\}$\;
\end{algorithm}

The algorithm is highly scalable. 
Indeed, the time complexity of Algorithm~\ref{alg:fast} is dominated by the part of computing the harmonic centrality scores of vertices in $N_\mathrm{in}(v)$, 
which just takes $O(|\rho(v)|(|V|+|A|))$ time. 
Therefore, the algorithm is asymptotically $b$ times faster than the greedy algorithm (Algorithm~\ref{alg:greedy}).

\subsection{Analysis}
From now on, we analyze the approximation ratio of Algorithm~\ref{alg:fast}. 
The following lemma demonstrates that the optimal value can be lower bounded 
using the maximum harmonic centrality score over the remaining incoming neighbors of $v$ in the resulting graph: 
\begin{lemma}\label{lem:fast_LB}
Let $F^*$ be an optimal solution to Problem~\ref{prob:centrality_min} 
and $N^*$ the vertex subset corresponding to $F^*$, i.e., $N^*=\{w\in V\mid (w,v)\in F^*\}$. 
Then it holds that 
\begin{align*}
f_{(G,v)}(F^*)\geq \frac{1}{2}\left(\max_{w\in N_\mathrm{in}(v)\setminus N^*}h_{G\setminus \rho(v)}(w)+1\right).
\end{align*}
\end{lemma}
\begin{proof}
For any $w\in N_\mathrm{in}(v)\setminus N^*$, we have 
\begin{align*}
&f_{(G,v)}(F^*)=h_{G\setminus F^*}(v)=\sum_{u\in V\setminus \{v\}}\frac{1}{d_{G\setminus F^*}(u,v)}\\
&\geq \sum_{u\in V\setminus \{v\}}\frac{1}{d_{G\setminus F^*}(u,w)+1}=1+\sum_{u\in V\setminus \{w,v\}}\frac{1}{d_{G\setminus F^*}(u,w)+1}\\
&\geq 1+\frac{1}{2}\sum_{u\in V\setminus \{w,v\}}\frac{1}{d_{G\setminus F^*}(u,w)}\geq \frac{1}{2}+\frac{1}{2}\sum_{u\in V\setminus \{w\}}\frac{1}{d_{G\setminus F^*}(u,w)}\\
&=\frac{1}{2}\left(h_{G\setminus F^*}(w)+1\right)
\geq \frac{1}{2}\left(h_{G\setminus \rho(v)}(w)+1\right), 
\end{align*}
where the first inequality follows from the triangle inequality of distance $d_{G\setminus F^*}$, 
the third equality follows from $d_{G\setminus F^*}(w,w)=0$, 
and the second inequality follows from the fact that the addition of $1$ to the denominator makes it at most twice the original. 
The arbitrariness of the choice of $w\in N_\mathrm{in}(v)\setminus N^*$ derives the statement. 
\end{proof}

On the other hand, the next lemma upper bounds the objective value of the output of Algorithm~\ref{alg:fast} using the harmonic centrality scores 
of the remaining incoming neighbors of $v$ in the resulting graph: 
\begin{lemma}\label{lem:fast_UB}
Let $F_\mathrm{ALG}$ be the output of Algorithm~\ref{alg:fast} 
and $N_\mathrm{ALG}$ the vertex subset corresponding to $F_\mathrm{ALG}$, i.e., $N_\mathrm{ALG}=\{w\in V\mid (w,v)\in F_\mathrm{ALG}\}$. 
Then we have 
\begin{align*}
f_{(G,v)}(F_\mathrm{ALG})\leq (|\rho(v)|-b)+\sum_{w\in N_\mathrm{in}(v)\setminus N_\mathrm{ALG}}h_{G\setminus \rho(v)}(w).
\end{align*}
\end{lemma}
\begin{proof}
On digraph $G\setminus F_\mathrm{ALG}$, any $u\in V\setminus \{v\}$ satisfies either 
(i) there exists no (shortest) path from $u$ to $v$ or 
(ii) there exists $w(u)\in N_\mathrm{in}(v)\setminus N_\mathrm{ALG}$ that is contained in a shortest path from $u$ to $v$, 
i.e., $d_{G\setminus F_\mathrm{ALG}}(u,v)=d_{G\setminus F_\mathrm{ALG}}(u,w(u))+1$. 
Let $V'\subseteq V\setminus \{v\}$ be the subset of vertices that satisfy the condition (ii). 
Then we have
\begin{align}\label{eq:fast_transform}
f_{(G,v)}(F_\mathrm{ALG})&=h_{G\setminus F_\mathrm{ALG}}(v)
=\sum_{u\in V\setminus \{v\}}\frac{1}{d_{G\setminus F_\mathrm{ALG}}(u,v)}\nonumber\\
&=\sum_{\substack{u\in V'\setminus \{v\}}}\frac{1}{d_{G\setminus F_\mathrm{ALG}}(u,w(u))+1}. 
\end{align}
We see that the shortest path corresponding to $d_{G\setminus F_\mathrm{ALG}}(u,w(u))$ does not contain $v$ (and thus any edge in $\rho(v)\setminus F_\mathrm{ALG}$). 
Otherwise there would exist $w'(u)\in N_\mathrm{in}(v)\setminus N_\mathrm{ALG}$ satisfying that $d_{G\setminus F_\mathrm{ALG}}(u,w'(u))<d_{G\setminus F_\mathrm{ALG}}(u,w(u))$, 
which contradicts the fact that $w(u)$ is contained in a shortest path from $u$ to $v$ on $G\setminus F_\mathrm{ALG}$. 
Hence, we have 
\begin{align*}
d_{G\setminus F_\mathrm{ALG}}(u,w(u))=d_{G\setminus \rho(v)}(u,w(u)). 
\end{align*}
Combining this with the equality~\eqref{eq:fast_transform}, we can conclude the proof: 
\begin{align*}
&f_{(G,v)}(F_\mathrm{ALG}) =\sum_{\substack{u\in V'\setminus \{v\}}}\frac{1}{d_{G\setminus \rho(v)}(u,w(u))+1}\\
&=(|\rho(v)|-b)+\sum_{\substack{u\in V'\setminus \{v\}\setminus (N_\mathrm{in}(v)\setminus N_\mathrm{ALG})}}\frac{1}{d_{G\setminus \rho(v)}(u,w(u))+1}\\
&\leq (|\rho(v)|-b)+\sum_{\substack{u\in V'\setminus \{v\}\setminus (N_\mathrm{in}(v)\setminus N_\mathrm{ALG})}}\frac{1}{d_{G\setminus \rho(v)}(u,w(u))}\\
&\leq (|\rho(v)|-b)+\sum_{w\in N_\mathrm{in}(v)\setminus N_\mathrm{ALG}}h_{G\setminus \rho(v)}(w),
\end{align*}
where the last inequality holds by the fact that 
any term $\frac{1}{d_{G\setminus \rho(v)}(u,w(u))}$ in the summation of the left-hand-side appears 
as a term in $h_{G\setminus \rho(v)}(w)$ for appropriate $w=w(u)$ in the right-hand-side. 
\end{proof}

We are now ready to prove our main theorem: 
\begin{theorem}\label{thm:fast_primal}
Algorithm~\ref{alg:fast} is a $2(|\rho(v)|-b)$-approximation algorithm for Problem~\ref{prob:centrality_min}. 
\end{theorem}

\begin{proof}
Here we use the notation that appeared in Lemmas~\ref{lem:fast_LB} and~\ref{lem:fast_UB}. 
By the behavior of Algorithm~\ref{alg:fast}, we have 
\begin{align*}
\max_{w\in N_\mathrm{in}(v)\setminus N_\mathrm{ALG}}h_{G\setminus \rho(v)}(w)\leq \max_{w\in N_\mathrm{in}(v)\setminus N^*}h_{G\setminus \rho(v)}(w). 
\end{align*}
Using Lemmas~\ref{lem:fast_LB} and~\ref{lem:fast_UB} together with this inequality, we have 
\begin{align*}
&f_{(G,v)}(F_\mathrm{ALG})\leq (|\rho(v)|-b)+\sum_{w\in N_\mathrm{in}(v)\setminus N_\mathrm{ALG}}h_{G\setminus \rho(v)}(w)\\
&\leq (|\rho(v)|-b)\left(1+\max_{w\in N_\mathrm{in}(v)\setminus N_\mathrm{ALG}} h_{G\setminus \rho(v)}(w)\right)\\
&\leq (|\rho(v)|-b)\left(1+\max_{w\in N_\mathrm{in}(v)\setminus N^*} h_{G\setminus \rho(v)}(w)\right)\\
&\leq (|\rho(v)|-b)\left(1+2f_{(G,v)}(F^*)-1\right)
= 2(|\rho(v)|-b)f_{(G,v)}(F^*), 
\end{align*}
which completes the proof. 
\end{proof}

Based on the theorem, we obtain the desired approximation ratio: 
\begin{corollary}
Algorithm~\ref{alg:fast} is a $\sqrt{2h_G(v)}$-approximation algorithm for Problem~\ref{prob:centrality_min}. 
\end{corollary}
\begin{proof}
For any instance that satisfies $b=|\rho(v)|$, Algorithm~\ref{alg:fast} outputs the trivial optimal solution (i.e., $\rho(v)$). 
Therefore, in what follows, we focus only on the instances with $b<|\rho(v)|$. 
Obviously the output of any algorithm for Problem~\ref{prob:centrality_min} has an objective value at most $h_G(v)$. 
On the other hand, the optimal value is at least $|\rho(v)|-b$ 
because in the resulting digraph, there are still $|\rho(v)|-b$ incoming neighbors of $v$, each of which contributes exactly $1$ to the objective value. 
Therefore, any algorithm (including Algorithm~\ref{alg:fast}) for Problem~\ref{prob:centrality_min} 
has an approximation ratio of $\frac{h_G(v)}{|\rho(v)|-b}$. 
By combining this with Theorem~\ref{thm:fast_primal}, the approximation ratio of Algorithm~\ref{alg:fast} can be improved to
$\min\left\{2(|\rho(v)|-b),\ \frac{h_G(v)}{|\rho(v)|-b}\right\}
\leq \sqrt{2h_G(v)}$. 
\end{proof}

We remark that as $\sqrt{2h_G(v)}\leq \sqrt{2(|V|-1)}=O(\sqrt{|V|})=o(|V|)$, 
the approximation ratio is unachievable by the greedy algorithm. 
Finally, we conclude this section by showing that the analysis of the approximation ratio is tight up to a constant factor. 
The proof is available in Appendix~\ref{appendix:proof_tight}. 
\begin{theorem}\label{thm:fast_limit}
Algorithm~\ref{alg:fast} has no approximation ratio of $o\left(\sqrt{h_G(v)}\right)$. 
\end{theorem}

\section{Bicriteria Approximation Algorithm}\label{sec:bicriteria}

In this section, we present a polynomial-time $(\frac{1}{\alpha},(\frac{1}{1-\alpha},\epsilon))$-bicriteria approximation algorithm ($\alpha\in (0,1)$ and $\epsilon >0$) for Problem~\ref{prob:centrality_min}.
Our algorithm first solves a continuous relaxation of the problem
and then applies a simple probabilistic procedure to the fractional solution to obtain the output.

\subsection{Continuous relaxation}
To obtain a continuous relaxation of Problem~\ref{prob:centrality_min},
we consider the well-known extension of set functions, called the Lov\'asz extension~\cite{Lovasz83}.
For our objective function $f_{(G,v)}$, the Lov\'asz extension $\widehat{f}_{(G,v)}\colon [0,1]^{\rho(v)}\rightarrow \mathbb{R}$ is defined in the following way:
Let $\rho(v)=\{e_1,\dots, e_{|\rho(v)|}\}$.
For $\bm{x}\in [0,1]^{\rho(v)}$, we relabel the elements of $\rho(v)$ so that $x_{e_1}\geq x_{e_2}\geq \cdots \geq x_{e_{|\rho(v)|}}$,
and construct a sequence of subsets $\emptyset=X_0\subset X_1\subset\cdots \subset X_{|\rho(v)|}=\rho(v)$,
where $X_i=\{e_1,\dots, e_i\}$ for $i=1,\dots,|\rho(v)|$.
Based on these, we define the value of $\widehat{f}_{(G,v)}(\bm{x})$ as follows:
\begin{align*}
&\widehat{f}_{(G,v)}(\bm{x}) = (1-x_{e_1})f_{(G,v)}(\emptyset)\\
&\qquad +\sum_{i=1}^{|\rho(v)|-1}(x_{e_i}- x_{e_{i+1}})f_{(G,v)}(X_i)+x_{e_{|\rho(v)|}}f_{(G,v)}(\rho(v)).
\end{align*}
Observe that for any $F\subseteq \rho(v)$, it holds that $\widehat{f}_{(G,v)}(\bm{1}_F)=f_{(G,v)}(F)$,
where $\bm{1}_F$ is an indicator vector of $F$, taking 1 if $e\in F$ and $0$ otherwise.
Therefore, $\widehat{f}_{(G,v)}$ is indeed an extension of $f_{(G,v)}$.

The Lov\'asz extension can be defined on any (not necessarily submodular) set function.
The Lov\'asz extension is always continuous but not necessarily differentiable.
An important fact is that the Lov\'asz extension is convex if and only if the original set function is submodular~\cite{Lovasz83}.
Therefore, by Theorem~\ref{thm:submodular}, the Lov\'asz extension $\widehat{f}_{(G,v)}$ of $f_{(G,v)}$ is convex.

Using $\widehat{f}_{(G,v)}$, we introduce our continuous relaxation as follows:
\begin{alignat*}{3}
&\textsf{Relaxation:} &\ \  &\text{minimize}   &\ \ &\widehat{f}_{(G,v)}(\bm{x})\\
&           &     &\text{subject to} &    &\|\bm{x}\|_1 \leq b, \ \bm{x}\in [0,1]^{\rho(v)}.
\end{alignat*}
For convenience, we denote by $C$ the feasible region of the problem, i.e., 
$C\coloneqq \left\{\bm{x}\in \mathbb{R}^{\rho(v)} : \|\bm{x}\|_1 \leq b\ \text{ and }\ \bm{x}\in [0,1]^{\rho(v)}\right\}$.
From the above, we see that \textsf{Relaxation} is a non-smooth convex programming problem.
We will present an algorithm for \textsf{Relaxation} and its convergence result in Section~\ref{sec:algorithm_cont}.
In the remainder of this section, we assume that for $\epsilon'=(1-\alpha)\epsilon$,
we can compute, in polynomial time, an $\epsilon'$-additive approximate solution for \textsf{Relaxation},
i.e., a feasible solution for \textsf{Relaxation} that has an objective value at most the optimal value plus $\epsilon'$.

\subsection{Algorithm}
Let $\bm{x}^*\in [0,1]^{\rho(v)}$ be an $\epsilon'$-additive approximate solution for \textsf{Relaxation}, where $\epsilon'=(1-\alpha)\epsilon$.
Then our algorithm picks $p\in [\alpha, 1]$ uniformly at random and just returns $\{e\in \rho(v)\mid x^*_e\geq p\}$.
For reference, the entire procedure is summarized in Algorithm~\ref{alg:biapprox}.

\begin{algorithm}[t]
\caption{$(\frac{1}{\alpha},(\frac{1}{1-\alpha},\epsilon))$-bicriteria approximation algorithm for Problem~\ref{prob:centrality_min}}
\label{alg:biapprox}
\SetKwInOut{Input}{Input}
\SetKwInOut{Output}{Output}
\Input{\ $G=(V,A)$, $v\in V$, and $b\in \mathbb{Z}_{>0}$}
\Output{\ $F\subseteq \rho(v)$}
$\epsilon'\leftarrow (1-\alpha)\epsilon$\;
Solve \textsf{Relaxation} (using Algorithm~\ref{alg:psm} in Section~\ref{sec:algorithm_cont}) and obtain its $\epsilon'$-additive approximate solution $\bm{x}^*\in [0,1]^{\rho(v)}$\;
Pick $p\in [\alpha, 1]$ uniformly at random\;
\Return $\{e\in \rho(v)\mid x^*_e\geq p\}$\;
\end{algorithm}

\subsection{Analysis}
The following theorem gives the bicriteria approximation ratio of Algorithm~\ref{alg:biapprox}: 

\begin{theorem}\label{thm:biapprox}
For any $\alpha \in (0,1)$ and $\epsilon >0$,
Algorithm~\ref{alg:biapprox} is a polynomial-time randomized $(\frac{1}{\alpha},(\frac{1}{1-\alpha},\epsilon))$-bicriteria approximation algorithm, in expectation, 
for Problem~\ref{prob:centrality_min}. 
\end{theorem}

\begin{proof}
Let $F\subseteq \rho(v)$ be the output of Algorithm~\ref{alg:biapprox}.
The approximation ratio with respect to the size of $F$ can be evaluated as follows:
\begin{align*}
\text{E}[|F|]&=\frac{1}{\alpha}\cdot \text{E}\left[\sum_{e\in F}\alpha\right]
\leq \frac{1}{\alpha}\cdot \text{E}\left[\sum_{e\in F}x^*_e\right]\leq \frac{1}{\alpha} \sum_{e\in \rho(v)}x^*_e\leq \frac{b}{\alpha},
\end{align*}
where the first inequality follows from $\alpha \leq p \leq x^*_e$ for any $e\in F$,
the second inequality follows from the nonnegativity of $\bm{x}^*$,
and the third inequality follows from the first constraint in \textsf{Relaxation}.

Next we analyze the approximation ratio with respect to the quality of $F$.
Let $F^*$ be an optimal solution to Problem~\ref{prob:centrality_min}.
As \textsf{Relaxation} is indeed a relaxation of Problem~\ref{prob:centrality_min} and $\bm{x}^*$ is its $\epsilon'$-additive approximate solution,
we have $\widehat{f}_{(G,v)}(\bm{x}^*)\leq f_{(G,v)}(F^*)+\epsilon'$.
For convenience, we define $x^*_{e_0}=1$ for an imaginary element $e_0$.
Let $\ell$ be the maximum number that satisfies $x^*_{e_\ell}\geq \alpha$.
Then we have
\begin{align*}
&\text{E}[f_{(G,v)}(F)]=\frac{\sum_{i=0}^{\ell-1}(x^*_{e_i}-x^*_{e_{i+1}})f_{(G,v)}(X_i)+ (x^*_{e_{\ell}}-\alpha)f_{(G,v)}(X_\ell)}{1-\alpha}\\
&\leq \frac{\sum_{i=0}^{|\rho(v)|-1}(x^*_{e_i}- x^*_{e_{i+1}})f_{(G,v)}(X_i)+x^*_{e_{|\rho(v)|}}f_{(G,v)}(\rho(v))}{1-\alpha}\\
&=\frac{\widehat{f}_{(G,v)}(\bm{x}^*)}{1-\alpha}\leq \frac{f_{(G,v)}(F^*)+\epsilon'}{1-\alpha}=\frac{f_{(G,v)}(F^*)}{1-\alpha}+\epsilon,
\end{align*}
where the first equality follows from the random choice of $p$ and the first inequality follows from the monotonicity of elements in $\bm{x}^*$ and the nonnegativity of $f_{(G,v)}$.
Therefore, we have the theorem.
\end{proof}

\section{Solving \textsf{Relaxation}}\label{sec:algorithm_cont}
In this section, we present our algorithm for solving \textsf{Relaxation}. 

\subsection{Algorithm}
Specifically, we design a projected subgradient method for \textsf{Relaxation}. 
The algorithm is an iterative method, where each iteration consists of two parts, i.e., the subgradient computation part and the projection computation part. 
The pseudo-code is given in Algorithm~\ref{alg:psm}. 
All the details will be given later. 
\begin{algorithm}[t]
\caption{Projected subgradient method for \textsf{Relaxation}}
\label{alg:psm}
\SetKwInOut{Input}{Input}
\SetKwInOut{Output}{Output}
\Input{\ $\bm{x}_0\in C$ and some stopping condition}
\Output{\ $\bm{x}\in C$}
$t\leftarrow 0$\;
\While{the stopping condition is not satisfied}{
  Pick a stepsize $\eta_t >0$ and a subgradient $\widehat{f}'_{(G,v)}(\bm{x}_{t})$ of $\widehat{f}_{(G,v)}$ at $\bm{x}_t$\;
  $\bm{x}_{t+1}\leftarrow \proj_C\left(\bm{x}_{t} - \eta_t \cdot \widehat{f}'(\bm{x}_{t})\right)$ and $t\leftarrow t+1$\;
}
\Return{$\bm{x}_t$}\;
\end{algorithm}
The sequence generated by the algorithm is $\{\bm{x}_t\}_{t\geq 0}$, while the sequence of function values generated by the algorithm is $\{\widehat{f}_{(G,v)}(\bm{x}^t)\}_{t\geq 0}$. 
As the sequence of function values is not necessarily monotone, we are also interested in the sequence of best-achieved function values at or before $\ell$-th iteration, 
which is defined as 
$\widehat{f}_\text{best}^{(\ell)}=\min_{t=0,1,\dots, \ell} \widehat{f}_{(G,v)}(\bm{x}_t)$. 

\noindent \textbf{Subgradient computation.}
From the definition of $\widehat{f}_{(G,v)}$, a subgradient $\widehat{f}'_{(G,v)}$ at $\bm{x}_t\in C$ is given by 
\begin{align}\label{eq:subgradient}
\widehat{f}'_{(G,v)}(\bm{x}_t)=\sum_{i=1}^{|\rho(v)|}\left(f_{(G,v)}(X_i)-f_{(G,v)}(X_{i-1})\right)\bm{u}_{e_i}, 
\end{align}
where $\bm{u}_{e_i}$ is the $|\rho(v)|$-dimensional vector that takes $1$ in the element corresponding to $e_i$ and $0$ elsewhere. 
To compute the subgradient $\widehat{f}'_{(G,v)}(\bm{x}_t)$,
we need to sort the entries of $\bm{x}_t$ and compute $f_{(G,v)}(X_i)$ for all $i=0,1,\dots, |\rho(v)|$,
which takes $O(|\rho(v)|(|V|+|A|))$ time. 

\smallskip
\noindent \textbf{Projection computation.}
For a given $\bm{x}\in \mathbb{R}^{\rho(v)}$, it is not trivial how to compute the projection of $\bm{x}$ onto $C$ 
because $C$ is the intersection of the two sets 
$\{\bm{x}\in \mathbb{R}^{\rho(v)} \mid \|\bm{x}\|_1 \leq b\}$ and $\{\bm{x}\in \mathbb{R}^{\rho(v)} \mid \bm{x}\in [0,1]^{\rho(v)}\}$. 
For simplicity, define $\mathrm{Box}[\bm{0},\bm{1}] = \{\bm{x}\in \mathbb{R}^{\rho(v)} \mid \bm{x}\in [0,1]^{\rho(v)}\}$. 
Let $\proj_{\mathrm{Box}[\bm{0},\bm{1}]}(\bm{x})$ be the projection of $\bm{x}$ onto $\mathrm{Box}[\bm{0},\bm{1}]$. 
Then by Lemma~6.26 in Beck~\cite{Beck17}, we obtain 
$\proj_{\mathrm{Box}[\bm{0},\bm{1}]}(\bm{x}) = (\min\{\max\{x_e,0\},1\})_{e\in \rho(v)}$. 
Using this projection, we can give the projection of $\bm{x}$ onto $C$ as follows: 
\begin{fact}[A special case of Example~6.32 in Beck~\cite{Beck17}]
Let $\proj_C(\bm{x})$ be the projection of $\bm{x}\in \mathbb{R}^{\rho(v)}$ onto $C$. Then we have 
\begin{align*}
\proj_C(\bm{x})=
\begin{cases}
\proj_{\mathrm{Box}[\bm{0},\bm{1}]}(\bm{x}) &\mathrm{if}\  \|\proj_{\mathrm{Box}[\bm{0},\bm{1}]}(\bm{x})\|_1\leq b,\\
\proj_{\mathrm{Box}[\bm{0},\bm{1}]}(\bm{x}-\lambda^*\bm{1}) &\mathrm{otherwise}, 
\end{cases}
\end{align*}
where $\lambda^*$ is any positive root of the nonincreasing function $\varphi(\lambda)=\|\proj_{\mathrm{Box}[\bm{0},\bm{1}]}(\bm{x}-\lambda\bm{1})\|_1-b$. 
\end{fact}

In practice, we can compute the value of $\lambda^*$ for $\bm{x}\coloneqq \bm{x}_{t} - \eta_t \cdot \widehat{f}'_{(G,v)}(\bm{x}_{t})$ using binary search. 
Assume that the stepsize $\eta_t >0$ is no greater than 1 for any iteration $t=0,1,\dots$, which is indeed the case of ours (specified later). 
As initial lower and upper bounds on $\lambda^*$, we can use $0$ and $\max_{e\in \rho(v)}x_e$, respectively. 
From the fact that $\bm{x}_{t}$ is always contained in $C$ and the definition of the subgradient~\eqref{eq:subgradient}, 
we see that 
$\max_{e\in \rho(v)}x_e
\leq 1+f_{(G,v)}(\emptyset)\leq |V|$. 
Therefore, the binary search finds $\lambda^*$ in $O(|\rho(v)|\log (|V|/\delta))$ time with an additive error of $\delta >0$. 
Note that any polynomial-time algorithm cannot recognize an additive error of $o(2^{-|V|^c})$ for constant $c$, due to its bit complexity. 
Hence, if we set $\delta =O(2^{-|V|^c})$, we can assume that the projection is exact, and the time complexity is still polynomial.

\subsection{Convergence result}
Let $L_{\widehat{f}_{(G,v)}}=\widehat{f}_{(G,v)}(\bm{0})$ ($=f_{(G,v)}(\emptyset)$). 
Based on the convergence result of the projected subgradient method in Beck~\cite{Beck17}, 
reviewed in Appendix~\ref{appendix:convergence}, we present the convergence result of Algorithm~\ref{alg:psm}: 
\begin{theorem}\label{thm:psm}
Let $\Theta$ be an upper bound on the half-squared diameter of $C$, i.e., $\Theta\geq \max_{\bm{x},\bm{y}\in C}\frac{1}{2} \|\bm{x}-\bm{y}\|^2$. 
Determine the stepsize $\eta_t$ ($t=0,1,\dots$) as $\eta_t=\frac{\sqrt{2\Theta}}{L_{\widehat{f}_{(G,v)}}\sqrt{t+1}}$.
Let $\widehat{f}^*$ be the optimal value of \textsf{Relaxation}. Then for all $t\geq 2$, it holds that 
\begin{align*}
\widehat{f}^{(t)}_\mathrm{best}-\widehat{f}^*\leq \frac{2(1+\log 3) L_{\widehat{f}_{(G,v)}} \sqrt{2\Theta}}{\sqrt{t+2}}. 
\end{align*}
\end{theorem}

The proof is in Appendix~\ref{appendix:proof_convergence}. 
By this theorem and the above discussion of the time complexity, the following is straightforward: 
\begin{corollary}
Let $\epsilon'>0$. Set the stopping condition of Algorithm~\ref{alg:psm} as follows: 
\begin{align*}
t\geq \left(\frac{2(1+\log 3)L_{\widehat{f}_{(G,v)}}\sqrt{2\Theta}}{\epsilon'}\right)^2-2. 
\end{align*}
Then, Algorithm~\ref{alg:psm} outputs, in polynomial time, an $\epsilon'$-additive approximate solution for \textsf{Relaxation}. 
\end{corollary}

\section{Experimental Evaluation}\label{sec:experiments}
In this section, we evaluate the performance of our proposed algorithms (i.e., Algorithms~\ref{alg:fast} and~\ref{alg:biapprox}) using various real-world networks.

\subsection{Setup}

\noindent \textbf{Instances.}
Table~\ref{tab:instance} lists real-world digraphs on which our experiments were conducted. 
All graphs were collected from the webpage of The KONECT Project\footnote{\url{http://konect.cc/}}. 
Note that self-loops and multiple edges were removed so that the graphs are made simple. 
For each graph, we randomly chose 20 vertices as target vertices among those having the in-degree at least $100$. 
The last column of Table~\ref{tab:instance} gives the statistics of the in-degrees of the target vertices, i.e., the maximum, average, and minimum in-degrees. 
For each graph and each target vertex $v$, we vary the budget $b$ in $\left\{\lfloor \frac{1}{4}|\rho(v)|\rfloor,\lfloor \frac{1}{2}|\rho(v)|\rfloor,\lfloor \frac{3}{4}|\rho(v)|\rfloor \right\}$. 
\begin{table}[t]
\begin{center}
\caption{Real-world graphs used in our experiments.}\label{tab:instance}
\scalebox{0.75}{
\begin{tabular}{lrrr}
\toprule
Name & $|V|$   & $|A|$   &Stat. of in-degrees of targets\\ 
\midrule
\texttt{moreno-blogs} &1,224 &19,022 & $(337,\, 158.80,\, 101)$\\
\texttt{dimacs10-polblogs} &1,224 & 33,430 & $(274,\, 143.45,\, 104)$\\
\texttt{librec-ciaodvd-trust} &4,658 &40,133 & $(361,\, 152.50,\, 100)$ \\
\texttt{munmun-twitter-social} &465,017 &834,797 & $(174,\, 123.65,\, 100)$\\
\texttt{citeseer} &384,054 &1,744,590 & $(495,\, 190.70,\, 106)$\\
\texttt{youtube-links} &1,138,494 &4,942,297 & $(\text{1,311},\, 273.40,\, 100)$\\
\texttt{higgs-twitter-social} &456,626 &14,855,819 & $(\text{1,049},\, 280.20,\, 111)$\\
\texttt{soc-pokec-relationships} &1,632,803 &30,622,564 & $(316,\, 157.35,\, 101)$\\
\bottomrule
\end{tabular}
}
\end{center}
\end{table}

\smallskip
\noindent \textbf{Baselines.} We employ the following baseline methods: 
\begin{itemize}
\leftskip=-5pt
\item \textsf{Empty}: This algorithm just outputs the empty set, thus presenting an upper bound on the objective function value (i.e., $h_G(v)$) of any feasible solution. 
\item  \textsf{Random}: This algorithm randomly chooses $b$ edges from $\rho(v)$. For each instance, this algorithm is run 100 times and the average objective value is reported. 
\item \textsf{Degree}: This algorithm sorts the elements of $N_\mathrm{in}(v)$ in the order of $(w_1,\dots w_{|\rho(v)|})$ so that $|\rho(w_1)|\geq \cdots \geq |\rho(w_{|\rho(v)|})|$ and just returns $\{(w_1,v),\dots, (w_b,v)\}$. 
\item \textsf{Greedy}: Execute the greedy algorithm (Algorithm~\ref{alg:greedy}). 
\end{itemize}

\noindent \textbf{Machine specs and code.}
All experiments were conducted on Mac mini with Apple M1 Chip and 16~GB RAM. 
All codes were written in Python 3.9, which are available online.\footnote{\url{https://github.com/atsushi-miyauchi/Local-Centrality-Minimization}}

\begin{figure*}
\includegraphics[width=0.245\textwidth]{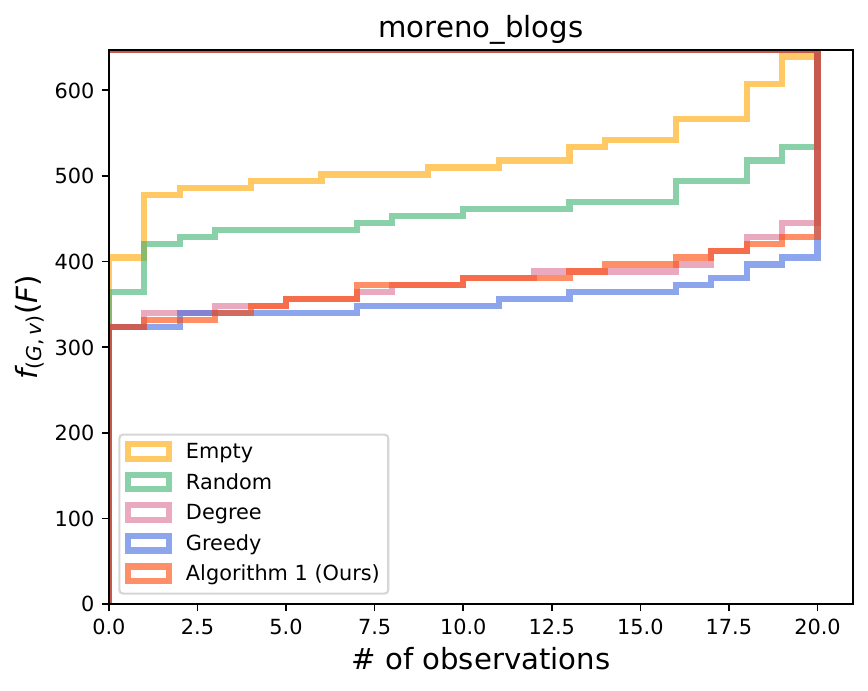}
\includegraphics[width=0.245\textwidth]{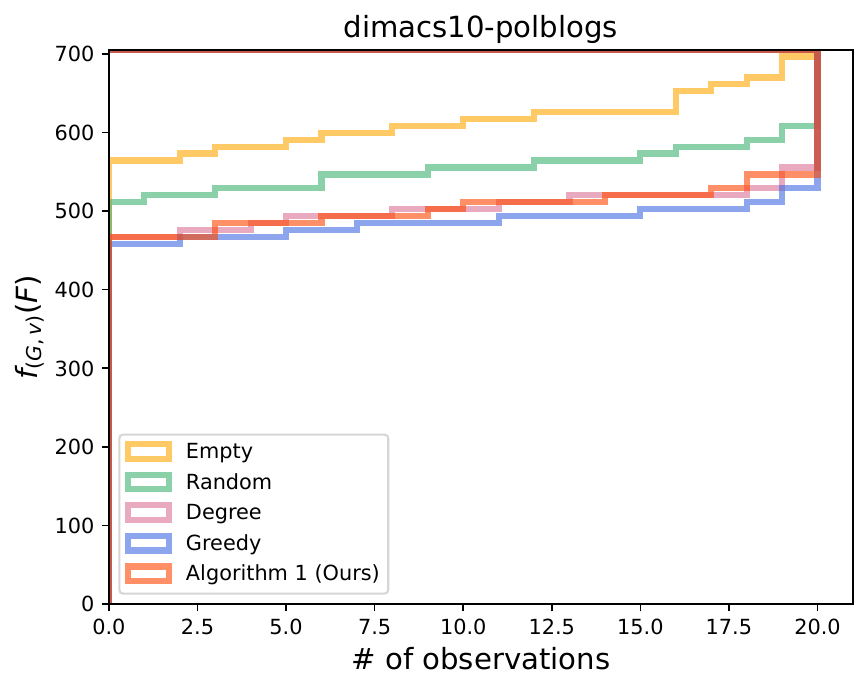}
\includegraphics[width=0.245\textwidth]{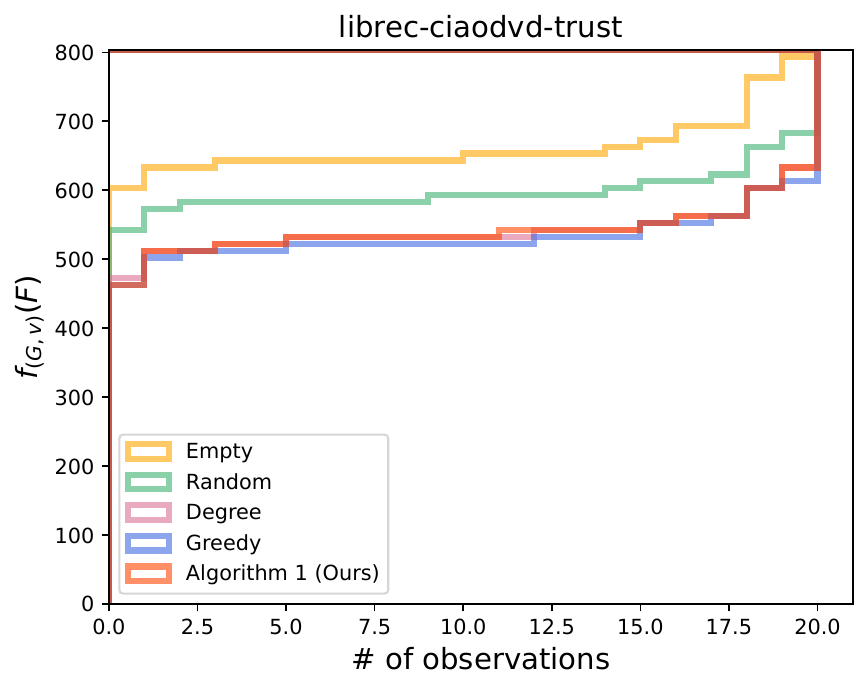}
\includegraphics[width=0.245\textwidth]{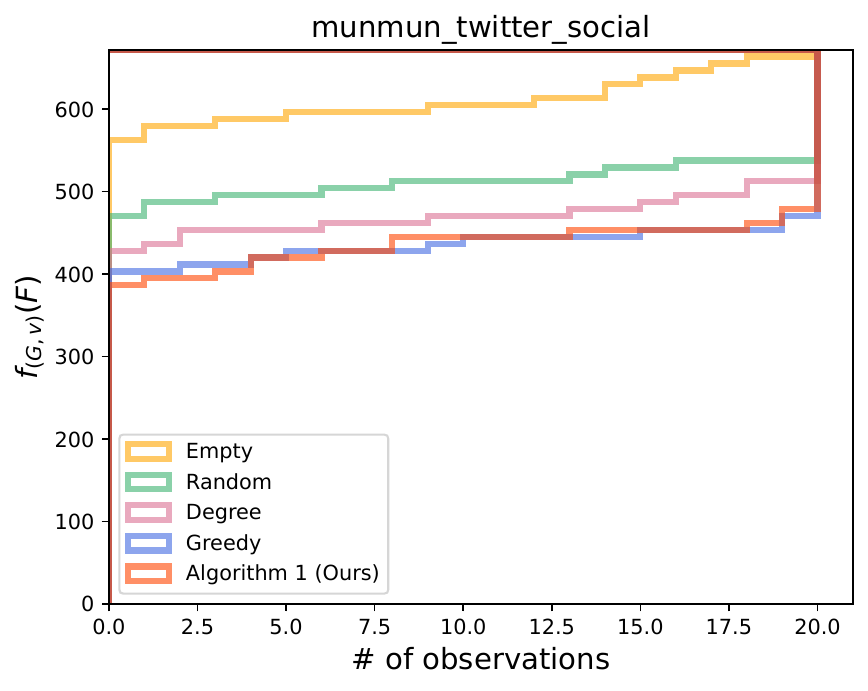}
\includegraphics[width=0.245\textwidth]{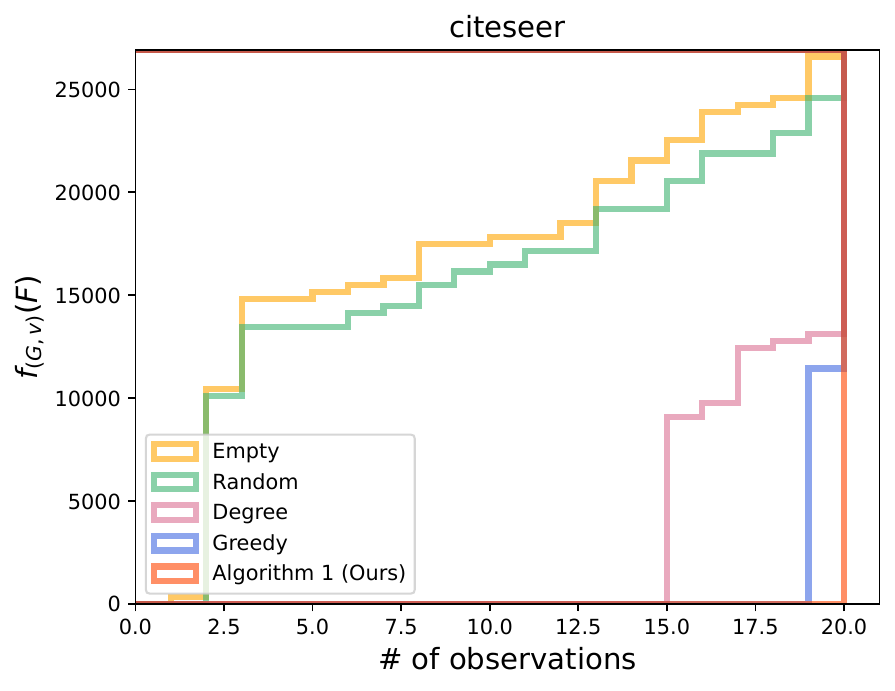}
\includegraphics[width=0.245\textwidth]{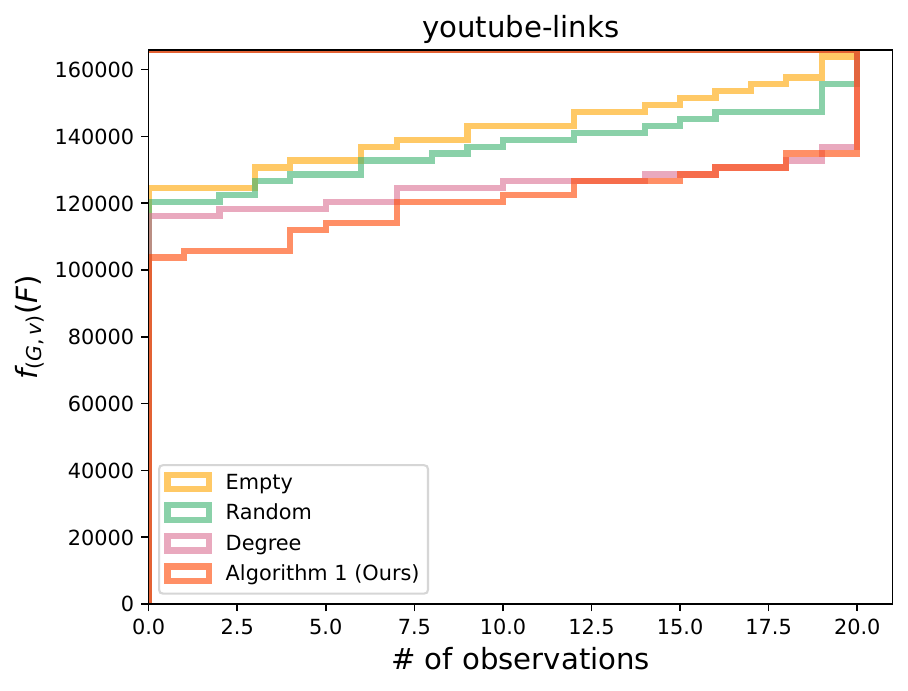}
\includegraphics[width=0.245\textwidth]{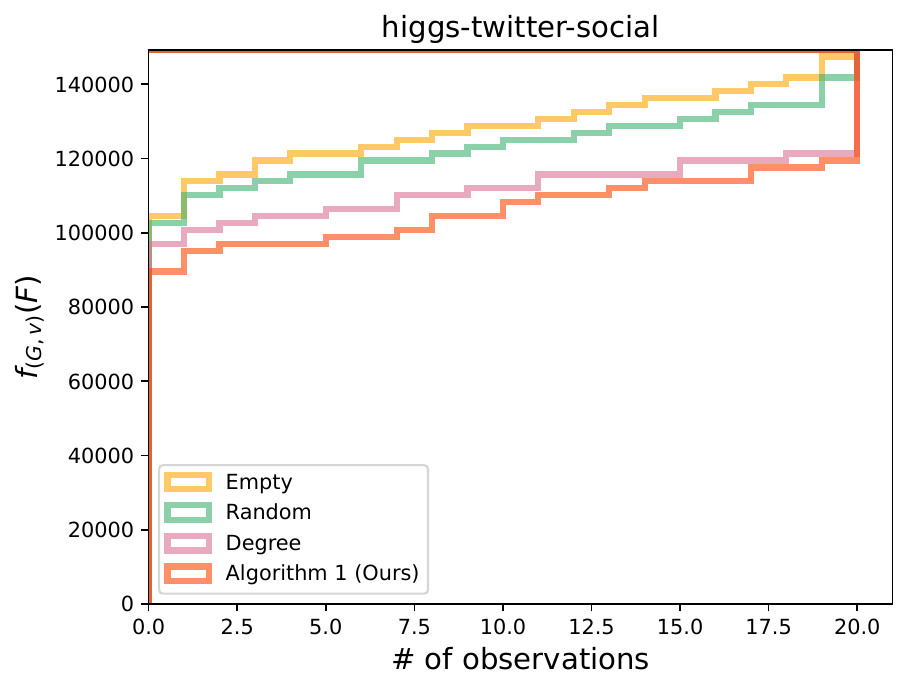}
\includegraphics[width=0.245\textwidth]{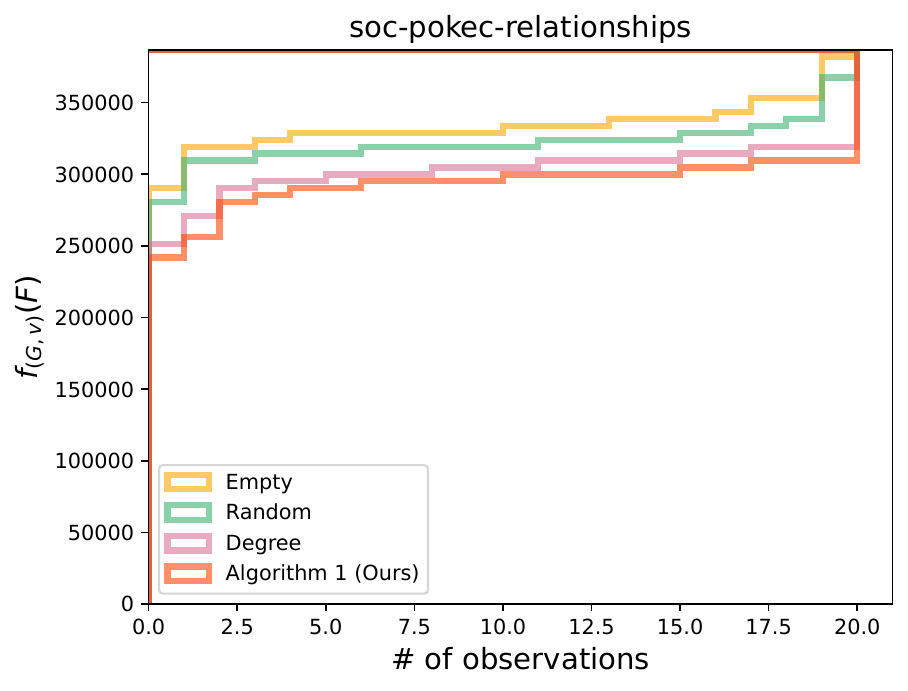}
\caption{Quality of solutions of the algorithms (except for Algorithm~\ref{alg:biapprox}) with $b=\lfloor \frac{1}{2}|\rho(v)|\rfloor$.}\label{fig:performance}
\end{figure*}

\subsection{Performance of algorithms}

Here we evaluate the performance of our algorithms. 
To this end, we run the algorithms together with the baselines for all graphs, target vertices, and budgets. 
For each graph and each budget, if the algorithm tested does not terminate within one hour for the target vertex having the largest in-degree, 
we do no longer run the algorithm for the other target vertices. 
Algorithm~\ref{alg:psm} (in Algorithm~\ref{alg:biapprox}) is run with stopping condition $t\geq 1000$ for scalability and initial solution $\bm{x}_0=\bm{0}$.

The quality of solutions of the algorithms except for Algorithm~\ref{alg:biapprox} is illustrated in Figure~\ref{fig:performance}. 
Due to space limitations, only the results for the budget $b=\lfloor \frac{1}{2}|\rho(v)|\rfloor$ are presented here. 
Although the trend is similar, the results for the other budget settings $b=\lfloor \frac{1}{4}|\rho(v)|\rfloor, \lfloor \frac{3}{4}|\rho(v)|\rfloor$ are given in Appendix~\ref{appendix:figures}. 
As the solutions of Algorithm~\ref{alg:biapprox} may violate the budget constraint, it is unfair to compare those with the others in the same plots; 
thus, the solutions are evaluated later. 
In the plots in Figure~\ref{fig:performance}, once we fix a value in the vertical axis, 
we can observe the cumulative number of solutions (i.e., targets) that attain the harmonic centrality score at most the fixed value. 
Therefore, we can say that an algorithm drawing a lower line has a better performance.

\begin{table}[t]
\begin{center}
\caption{Computation time (seconds) of the algorithms tested.}\label{tab:time}
\scalebox{0.75}{
\begin{tabular}{lcrrrrrrrrrrrrrrrrrr}
\toprule
Name & $b$ &
\textsf{Greedy}&
Algorithm~\ref{alg:fast}&
Algorithm~\ref{alg:biapprox}
\\
\midrule
& $\lfloor \frac{1}{4}|\rho(v)| \rfloor$ &4.53 & 0.10 &326.55 \\
\texttt{moreno-blogs} & $\lfloor \frac{1}{2}|\rho(v)| \rfloor$ &7.79 & 0.10 &333.48 \\
& $\lfloor \frac{3}{4}|\rho(v)| \rfloor$ &9.75 & 0.10 &332.48 \\
\midrule
& $\lfloor \frac{1}{4}|\rho(v)| \rfloor$ & 4.40 & 0.12 & 393.72\\
\texttt{dimacs10-polblogs} & $\lfloor \frac{1}{2}|\rho(v)| \rfloor$ & 7.59 & 0.13 &416.21 \\
& $\lfloor \frac{3}{4}|\rho(v)| \rfloor$ & 9.55 & 0.13 &404.52 \\
\midrule
& $\lfloor \frac{1}{4}|\rho(v)| \rfloor$ & 6.17 & 0.15 &482.83 \\
\texttt{librec-ciaodvd-trust} & $\lfloor \frac{1}{2}|\rho(v)| \rfloor$ & 10.59 & 0.15 &502.63 \\
& $\lfloor \frac{3}{4}|\rho(v)| \rfloor$ & 13.33 & 0.15 & 489.92 \\
\midrule
& $\lfloor \frac{1}{4}|\rho(v)| \rfloor$ &4.03 &0.14 &440.70  \\
\texttt{munmun-twitter-social} & $\lfloor \frac{1}{2}|\rho(v)| \rfloor$ &6.90 &0.14 &444.37 \\
& $\lfloor \frac{3}{4}|\rho(v)| \rfloor$ &8.68 &0.14 &446.68  \\
\midrule
& $\lfloor \frac{1}{4}|\rho(v)| \rfloor$ &1,405.08 &4.59 &--- \\
\texttt{citeseer} & $\lfloor \frac{1}{2}|\rho(v)| \rfloor$ &1,671.77 &4.55 &---\\
& $\lfloor \frac{3}{4}|\rho(v)| \rfloor$  &1,662.66 &4.55 &---\\
\midrule
& $\lfloor \frac{1}{4}|\rho(v)| \rfloor$ &--- &130.56 &---\\
\texttt{youtube-links} & $\lfloor \frac{1}{2}|\rho(v)| \rfloor$ &--- &130.54 &---\\
& $\lfloor \frac{3}{4}|\rho(v)| \rfloor$ &--- &132.19 &--- \\
\midrule
& $\lfloor \frac{1}{4}|\rho(v)| \rfloor$ &--- &188.84 &---\\
\texttt{higgs-twitter-social} & $\lfloor \frac{1}{2}|\rho(v)| \rfloor$ &--- &187.92 &---\\
& $\lfloor \frac{3}{4}|\rho(v)| \rfloor$ &--- &187.68 &--- \\
\midrule
& $\lfloor \frac{1}{4}|\rho(v)| \rfloor$ &--- &402.24 &---\\
\texttt{soc-pokec-relationships} & $\lfloor \frac{1}{2}|\rho(v)| \rfloor$ &--- &392.63 &---\\
& $\lfloor \frac{3}{4}|\rho(v)| \rfloor$ &--- &512.05 &--- \\
\bottomrule
\end{tabular}
}
\end{center}
\end{table}

As can be seen, Algorithm~\ref{alg:fast} outperforms the baselines. 
Indeed, Algorithm~\ref{alg:fast} is applicable to all graphs tested, thanks to its high scalability, 
and the quality of solutions is much better than that of the scalable baselines, i.e., \textsf{Random} and \textsf{Degree}. 
Most interestingly, for the graph \texttt{citeseer}, 
Algorithm~\ref{alg:fast} succeeds in reducing the harmonic centrality scores of all target vertices to the values relatively close to $0$, 
although \textsf{Degree} and \textsf{Greedy} fail to have centrality scores less than 9,000 for some target vertices. 
Note that this result does not contradict the fact that even an optimal solution has the harmonic centrality score no less than $|\rho(v)|-b$: 
The minimum objective value attained by Algorithm~\ref{alg:fast} is $54$, rather than $0$. 
For small graphs, \textsf{Greedy} performs slightly better than Algorithm~\ref{alg:fast}; 
however, as \textsf{Greedy} is quite time consuming, 
it is not applicable to \texttt{youtube-links} and the larger graphs. 

The detailed report of the computation time is found in Table~\ref{tab:time}, 
where the average computation time over all target vertices is reported. 
Note that the results for \textsf{Random} and \textsf{Degree} are omitted 
because \textsf{Random} is obviously quite fast and \textsf{Degree} records 0.00 seconds for all graphs and budgets. 
We remark that the computation time of \textsf{Greedy} grows roughly proportionally to the budget $b$, 
but that of Algorithm~\ref{alg:fast} remains almost the same for all settings of $b$. 

Finally, the quality of solutions of Algorithm~\ref{alg:biapprox} (with $\alpha=\frac{1}{3},\frac{1}{2}$) for $b=\lfloor \frac{1}{4}|\rho(v)|\rfloor$, averaged over the target vertices, is reported in Table~\ref{tab:biapprox}, 
where each of the objective value and the size of the solutions is a relative value compared with that of the solutions of Algorithm~\ref{alg:fast}. 
As the rounding procedure of the algorithm contains randomness, we run the procedure 100 times and report the average value. 
As can be seen, Algorithm~\ref{alg:biapprox} returns very small solutions that even do not exhaust the budget. 
However, this result does not contradict the guarantee given in Theorem~\ref{thm:biapprox}; 
the objective value is not much worse than that of Algorithm~\ref{alg:fast} and the (expected) size of solutions is just upper bounded in Theorem~\ref{thm:biapprox}. 

\begin{table}[t]
\begin{center}
\caption{Quality of solutions of Algorithm~\ref{alg:biapprox} with $\alpha=\frac{1}{3},\frac{1}{2}$ compared with those of Algorithm~\ref{alg:fast} with $b=\lfloor \frac{1}{4}|\rho(v)|\rfloor$.}\label{tab:biapprox}
\scalebox{0.75}{
\begin{tabular}{lrrrr}
\toprule
Name & \multicolumn{2}{c}{$\alpha=\frac{1}{3}$} &\multicolumn{2}{c}{$\alpha=\frac{1}{2}$}\\
\cmidrule(lr){2-3}
\cmidrule(lr){4-5}
     & Obj. val. & Size & Obj. val. & Size \\
\midrule
\texttt{moreno-blogs} &1.12 &0.14 &1.13 &0.12 \\
\texttt{dimacs10-polblogs} &1.09 &0.10 &1.09 &0.07 \\
\texttt{librec-ciaodvd-trust} &1.09 &0.06 &1.09 &0.03\\
\texttt{munmun-twitter-social} &1.15 &0.08 &1.15 &0.08\\
\bottomrule
\end{tabular}
}
\end{center}
\end{table}

To further examine the performance of Algorithm~\ref{alg:biapprox}, 
we run it on the instance that appeared in the proof of Theorem~\ref{thm:fast_limit}, where we set $k=50$ (and $b=50$). 
Note that it is theoretically guaranteed that Algorithm~\ref{alg:fast} outputs a poor solution to this instance. 
On the other hand, Algorithm~\ref{alg:biapprox} with $\alpha=\frac{3}{4}$ always obtains the optimal solution among 100 rounding trials. 
Hence, we see that the algorithm is rather strong against adversarial instances, verifying the effectiveness of the theoretical approximation guarantee of Algorithm~\ref{alg:biapprox}.

\section{Conclusion}\label{sec:conclusion}

In this study, we have introduced a novel optimization model for local centrality minimization 
and designed two effective approximation algorithms. 
To our knowledge, ours are the first polynomial-time algorithms with provable approximation guarantees for centrality minimization. 
Experiments using a variety of real-world networks demonstrate the effectiveness of our proposed algorithms. 

Our work opens up several interesting problems. 
Can we design a polynomial-time algorithm that has an approximation ratio better than that of Algorithm~\ref{alg:fast} or a bicriteria approximation ratio better than that of Algorithm~\ref{alg:biapprox}?
Another interesting direction is to study Problem~\ref{prob:centrality_min} with a more capable setting. 
For example, it would be valuable to have a target vertex subset (rather than a single target vertex) and aim to minimize its group harmonic centrality score, as in the literature on global centrality minimization~\cite{Veremyev+19}.

\clearpage
\begin{acks}
Lorenzo Severini was supported for this research by Project ECS 0000024 Rome Technopole - CUP B83C22002820006, ``PNRR Missione 4 Componente 2 Investimento 1.5'', funded by European Commission - NextGenerationEU.
\end{acks}

\bibliographystyle{abbrv}
\balance
\bibliography{main}

\begin{thebibliography}{10}

\bibitem{Angriman+21}
E.~Angriman, R.~Becker, G.~D'Angelo, H.~Gilbert, A.~van~der Grinten, and
  H.~Meyerhenke.
\newblock Group-harmonic and group-closeness maximization -- {A}pproximation
  and engineering.
\newblock In {\em Proceedings of ALENEX~'21}, pages 154--168, 2021.

\bibitem{Angriman+20}
E.~Angriman, A.~van~der Grinten, A.~Bojchevski, D.~Zügner, S.~Günnemann, and
  H.~Meyerhenke.
\newblock Group centrality maximization for large-scale graphs.
\newblock In {\em Proceedings of ALENEX~'20}, pages 56--69, 2020.

\bibitem{Beck17}
A.~Beck.
\newblock {\em First-Order Methods in Optimization}.
\newblock MOS--SIAM Series on Optimization. SIAM, 2017.

\bibitem{Bergamini+18}
E.~Bergamini, P.~Crescenzi, G.~D'angelo, H.~Meyerhenke, L.~Severini, and
  Y.~Velaj.
\newblock Improving the betweenness centrality of a node by adding links.
\newblock {\em ACM Journal of Experimental Algorithmics}, 23:1--32, 2018.

\bibitem{Bhawalkar+15}
K.~Bhawalkar, J.~Kleinberg, K.~Lewi, T.~Roughgarden, and A.~Sharma.
\newblock Preventing unraveling in social networks: {T}he anchored $k$-core
  problem.
\newblock {\em SIAM Journal on Discrete Mathematics}, 29(3):1452--1475, 2015.

\bibitem{Boldi_Vigna_14}
P.~Boldi and S.~Vigna.
\newblock Axioms for centrality.
\newblock {\em Internet Mathematics}, 10(3--4):222--262, 2014.

\bibitem{Castaldo+20}
M.~Castaldo, C.~Catalano, G.~Como, and F.~Fagnani.
\newblock On a centrality maximization game.
\newblock {\em IFAC-PapersOnLine}, 53(2):2844--2849, 2020.

\bibitem{Chen+16}
C.~Chen, W.~Wang, and X.~Wang.
\newblock Efficient maximum closeness centrality group identification.
\newblock In {\em Proceedings of the 27th Australasian Database Conference},
  pages 43--55, 2016.

\bibitem{Coupette+23}
C.~Coupette, S.~Neumann, and A.~Gionis.
\newblock Reducing exposure to harmful content via graph rewiring.
\newblock In {\em Proceedings of KDD~'23}, pages 323--334, 2023.

\bibitem{Crescenzi+16}
P.~Crescenzi, G.~D'angelo, L.~Severini, and Y.~Velaj.
\newblock Greedily improving our own closeness centrality in a network.
\newblock {\em ACM Transactions on Knowledge Discovery from Data}, 11(1):1--32,
  2016.

\bibitem{DAngelo+19}
G.~D'Angelo, M.~Olsen, and L.~Severini.
\newblock Coverage centrality maximization in undirected networks.
\newblock In {\em Proceedings of AAAI~'19}, pages 501--508, 2019.

\bibitem{Das+18}
K.~Das, S.~Samanta, and M.~Pal.
\newblock Study on centrality measures in social networks: A survey.
\newblock {\em Social Network Analysis and Mining}, 8:1--11, 2018.

\bibitem{Fabbri+22}
F.~Fabbri, Y.~Wang, F.~Bonchi, C.~Castillo, and M.~Mathioudakis.
\newblock Rewiring what-to-watch-next recommendations to reduce radicalization
  pathways.
\newblock In {\em Proceedings of The Web Conference~'22}, page 2719–2728,
  2022.

\bibitem{Fujishige05}
S.~Fujishige.
\newblock {\em Submodular Functions and Optimization}, volume~58 of {\em Annals
  of Discrete Mathematics}.
\newblock Elsevier, 2005.

\bibitem{Goldenberg+01_1}
J.~Goldenberg, B.~Libai, and E.~Muller.
\newblock Talk of the network: A complex systems look at the underlying process
  of word-of-mouth.
\newblock {\em Marketing Letters}, 12:211--223, 2001.

\bibitem{Goldenberg+01_2}
J.~Goldenberg, B.~Libai, and E.~Muller.
\newblock Using complex systems analysis to advance marketing theory
  development: Modeling heterogeneity effects on new product growth through
  stochastic cellular automata.
\newblock {\em Academy of Marketing Science Review}, 9(3):1--18, 2001.

\bibitem{Hayrapetyan+_05}
A.~Hayrapetyan, D.~Kempe, M.~P{\'a}l, and Z.~Svitkina.
\newblock Unbalanced graph cuts.
\newblock In {\em Proceedings of ESA~'05}, pages 191--202, 2005.

\bibitem{Ishakian+12}
V.~Ishakian, D.~Erdös, E.~Terzi, and A.~Bestavros.
\newblock A framework for the evaluation and management of network centrality.
\newblock In {\em Proceedings of SDM~'12}, pages 427--438, 2012.

\bibitem{Kempe+03}
D.~Kempe, J.~Kleinberg, and E.~Tardos.
\newblock Maximizing the spread of influence through a social network.
\newblock In {\em Proceedings of KDD~'03}, pages 137--146, 2003.

\bibitem{Khalil+13}
E.~Khalil, B.~Dilkina, and L.~Song.
\newblock Cutting{E}dge: Influence minimization in networks.
\newblock In {\em Proceedings of Workshop on Frontiers of Network Analysis:
  Methods, Models, and Applications}, pages 1--13, 2013.

\bibitem{Knoke+20}
D.~Knoke and S.~Yang.
\newblock {\em Social Network Analysis}, volume 154 of {\em Quantitative
  Applications in the Social Sciences}.
\newblock SAGE Publications, 2020.

\bibitem{Konar+21}
A.~Konar and N.~D. Sidiropoulos.
\newblock Exploring the subgraph density-size trade-off via the {L}ova\'sz
  extension.
\newblock In {\em Proceedings of WSDM~'21}, pages 743--751, 2021.

\bibitem{Konar+22}
A.~Konar and N.~D. Sidiropoulos.
\newblock The triangle-densest-$k$-subgraph problem: Hardness, {L}ov{\'a}sz
  extension, and application to document summarization.
\newblock In {\em Proceedings of AAAI~'22}, pages 4075--4082, 2022.

\bibitem{Li+19}
H.~Li, R.~Peng, L.~Shan, Y.~Yi, and Z.~Zhang.
\newblock Current flow group closeness centrality for complex networks?
\newblock In {\em Proceedings of The Web Conference~'19}, pages 961--971, 2019.

\bibitem{Liu+23}
C.~Liu, X.~Zhou, A.~N. Zehmakan, and Z.~Zhang.
\newblock A fast algorithm for moderating critical nodes via edge removal.
\newblock {\em IEEE Transactions on Knowledge and Data Engineering}, pages
  1--14, 2023.
\newblock In press.

\bibitem{Lovasz83}
L.~Lov{\'a}sz.
\newblock Submodular functions and convexity.
\newblock In A.~Bachem, B.~Korte, and M.~Gr\"otschel, editors, {\em
  Mathematical Programming: The State of the Art}, pages 235--257. Springer,
  1983.

\bibitem{Mahmoody+16}
A.~Mahmoody, C.~E. Tsourakakis, and E.~Upfal.
\newblock Scalable betweenness centrality maximization via sampling.
\newblock In {\em Proceedings of KDD~'16}, pages 1765--1773, 2016.

\bibitem{Maulana+21}
A.~Maulana and M.~Atzmueller.
\newblock Many-objective optimization for anomaly detection on multi-layer
  complex interaction networks.
\newblock {\em Applied Sciences}, 11(9):4005, 2021.

\bibitem{Medya+18}
S.~Medya, A.~Silva, A.~Singh, P.~Basu, and A.~Swami.
\newblock Group centrality maximization via network design.
\newblock In {\em Proceedings of SDM~'18}, pages 126--134, 2018.

\bibitem{Minoux05}
M.~Minoux.
\newblock Accelerated greedy algorithms for maximizing submodular set
  functions.
\newblock In {\em Proceedings of the 8th IFIP Conference on Optimization
  Techniques}, pages 234--243, 2005.

\bibitem{Mumtaz+17}
S.~Mumtaz and X.~Wang.
\newblock Identifying top-k influential nodes in networks.
\newblock In {\em Proceedings of CIKM~'17}, pages 2219--2222, 2017.

\bibitem{Murai_Yoshida_19}
S.~Murai and Y.~Yoshida.
\newblock Sensitivity analysis of centralities on unweighted networks.
\newblock In {\em Proceedings of The Web Conference~'19}, pages 1332--1342,
  2019.

\bibitem{Nemhauser+78}
G.~L. Nemhauser and L.~A. Wolsey.
\newblock Best algorithms for approximating the maximum of a submodular set
  function.
\newblock {\em Mathematics of Operations Research}, 3(3):177--188, 1978.

\bibitem{Newman+04}
M.~E.~J. Newman and M.~Girvan.
\newblock Finding and evaluating community structure in networks.
\newblock {\em Physical Review E}, 69:026113, 2004.

\bibitem{Pellegrina23}
L.~Pellegrina.
\newblock Efficient centrality maximization with {R}ademacher averages.
\newblock In {\em Proceedings of KDD~'23}, pages 1872--1884, 2023.

\bibitem{Rangapuram+13}
S.~S. Rangapuram, T.~B\"{u}hler, and M.~Hein.
\newblock Towards realistic team formation in social networks based on densest
  subgraphs.
\newblock In {\em Proceedings of WWW~'13}, pages 1077--1088, 2013.

\bibitem{Svitkina_Fleischer_11}
Z.~Svitkina and L.~Fleischer.
\newblock Submodular approximation: Sampling-based algorithms and lower bounds.
\newblock {\em SIAM Journal on Computing}, 40(6):1715--1737, 2011.

\bibitem{Veremyev+19}
A.~Veremyev, O.~A. Prokopyev, and E.~L. Pasiliao.
\newblock Finding critical links for closeness centrality.
\newblock {\em INFORMS Journal on Computing}, 31(2):367--389, 2019.

\bibitem{Vinterbo_04}
S.~A. Vinterbo.
\newblock Privacy: {A} machine learning view.
\newblock {\em IEEE Transactions on Knowledge and Data Engineering},
  16(8):939--948, 2004.

\bibitem{Wang+20}
X.~Wang, K.~Deng, J.~Li, J.~X. Yu, C.~S. Jensen, and X.~Yang.
\newblock Efficient targeted influence minimization in big social networks.
\newblock {\em World Wide Web}, 23(4):2323--2340, 2020.

\bibitem{Waniek+18}
M.~Waniek, T.~P. Michalak, M.~J. Wooldridge, and T.~Rahwan.
\newblock Hiding individuals and communities in a social network.
\newblock {\em Nature Human Behaviour}, 2(2):139, 2018.

\bibitem{Waniek+23}
M.~Waniek, J.~Woźnica, K.~Zhou, Y.~Vorobeychik, T.~P. Michalak, and T.~Rahwan.
\newblock Hiding from centrality measures: A {S}tackelberg game perspective.
\newblock {\em IEEE Transactions on Knowledge and Data Engineering},
  35(10):10058--10071, 2023.

\bibitem{Yang+19}
L.~Yang, Z.~Li, and A.~Giua.
\newblock Influence minimization in linear threshold networks.
\newblock {\em Automatica}, 100:10--16, 2019.

\bibitem{Zhang+17}
F.~Zhang, Y.~Zhang, L.~Qin, W.~Zhang, and X.~Lin.
\newblock Finding critical users for social network engagement: The collapsed
  k-core problem.
\newblock In {\em Proceedings of AAAI~'17}, 2017.

\bibitem{Zhao+17}
P.~Zhao, Y.~Li, H.~Xie, Z.~Wu, Y.~Xu, and J.~C. Lui.
\newblock Measuring and maximizing influence via random walk in social activity
  networks.
\newblock In {\em Proceedings of DASFAA~'17}, pages 323--338, 2017.

\bibitem{Zhu+18}
W.~Zhu, C.~Chen, X.~Wang, and X.~Lin.
\newblock K-core minimization: {A}n edge manipulation approach.
\newblock In {\em Proceedings of CIKM~'18}, pages 1667--1670, 2018.

\end{thebibliography}

\appendix
\section{Omitted Contents from Section~\ref{sec:problem}}

\subsection{Proof of Theorem~\ref{thm:hardness}}\label{appendix:proof_hardness}

To prove the theorem, we formally introduce the following problem:

\begin{problem}[Minimum $k$-union]\label{prob:min_k-union}
Given a ground set $U=\{e_1,\dots,e_n\}$, a set system $\mathcal{S}=\{S_1,\dots, S_m\}\subseteq 2^{U}$, and $k\in \mathbb{Z}_{>0}$,
we are asked to find $J\subseteq \{1,\dots, m\}$ with $|J|=k$ that minimizes $\left| \bigcup_{j\in J} S_j\right|$.
Without loss of generality, it can be assumed that $\bigcup_{j\in \{1,\dots, m\}}S_j=U$.
\end{problem}

\begin{fact}[Theorem~1 in Vinterbo~\cite{Vinterbo_04}]
Problem~\ref{prob:min_k-union} is NP-hard.
\end{fact}

\begin{proof}[Proof of Theorem~\ref{thm:hardness}]
We construct a polynomial-time reduction from Problem~\ref{prob:min_k-union} to Problem~\ref{prob:centrality_min}.
Take an arbitrary instance of Problem~\ref{prob:min_k-union}: $U=\{e_1,\dots,e_n\}$, $\mathcal{S}=\{S_1,\dots, S_m\}\subseteq 2^{U}$, and $k\in \mathbb{Z}_{>0}$.
We make an instance (i.e., a gadget) of Problem~\ref{prob:centrality_min}, i.e., $G=(V,A)$, $v\in V$, and $b\in \mathbb{Z}_{>0}$, as follows:
\begin{itemize}
\item $G=(V,A)$ such that
\begin{itemize}
\item $V=\{v_0, v_{S_1}, \dots, v_{S_m}, v_{e_1}, \dots, v_{e_n}\}$,
\item $A=\{(v_{S_j},v_0)\mid j=1,\dots, m\}\cup \{(v_{e_i},v_{S_j})\mid  e_i\in S_j,\, i=1,\dots, n,\, j=1,\dots, m\}$;
\end{itemize}
\item $v=v_0$;
\item $b=m-k$.
\end{itemize}
Clearly, $G$ is acyclic.
Let $F^* \subseteq \rho(v_0)\ (=\{(v_{S_j},v_0)\mid j=1,\dots, m\})$ be an optimal solution to this gadget,
and let $J^*\subseteq \{1,\dots, m\}$ be the index set that satisfies $F^*=\{(v_{S_j},v_0)\mid j\in J^*\}$.
We see that $|J^*|=b$.
Let $\overline{J^*}=\{1,\dots, m\}\setminus J^*$.
Then, the objective value of $F^*$, i.e., the harmonic centrality of $v_0$ after the removal of $F^*$, can be evaluated as
$(m-b) + \frac{1}{2}\left|\bigcup_{j\in \overline{J^*}}S_j\right|$.
Therefore, as $F^*$ is optimal, $\left|\bigcup_{j\in \overline{J^*}}S_j\right|$ is minimized.
Noticing that $|\overline{J^*}|=m-b=m-(m-k)=k$, we see that $\overline{J^*}$ is an optimal solution to the instance of Problem~\ref{prob:min_k-union}.
\end{proof}

\subsection{Proof of Theorem~\ref{thm:submodular}}\label{appendix:proof_submodular}

Before proving the theorem, we mention a fundamental fact of the submodularity.
It is well known that the submodularity is equivalent to the diminishing marginal return property~\cite{Fujishige05}, which can be written as follows:
For any $X\subset Y \subseteq S$ and any $p\in S\setminus Y$, it holds that
\begin{align*}
f(X\cup \{p\}) - f(X) \geq f(Y\cup \{p\}) - f(Y).
\end{align*}

\begin{proof}[Proof of Theorem~\ref{thm:submodular}]
Let $G=(V,A)$ be any digraph and $v\in V$ be any vertex.
From the above fact, it suffices to show that for any $E\subset F\subseteq \rho(v)$ and $e\in \rho(v)\setminus F$,
\begin{align*}
\obj(E\cup \{e\})-\obj(E)\geq \obj(F\cup \{e\})- \obj(F)
\end{align*}
holds.
From the definition of $\obj$, we have
\begin{align*}
&\obj(E\cup \{e\})-\obj(E)\\
&=h_{G\setminus (E\cup \{e\})}(v) - h_{G\setminus E}(v)\\
&=\sum_{u\in V\setminus \{v\}}\left(\frac{1}{d_{G\setminus (E\cup \{e\})}(u,v)} -\frac{1}{d_{G\setminus E}(u,v)}\right).
\end{align*}
For each $e\in \rho(v)$, we define
\begin{align*}
U_G(v, e)=\{&u\in V\setminus \{v\} \mid \\
&\text{there exists a path from $u$ to $v$ on $G$ and}\\
&\text{all shortest paths from $u$ to $v$ on $G$ contain $e$}\}.
\end{align*}
Noticing that $d_{G\setminus (E\cup \{e\})}(u,v)= d_{G\setminus E}(u,v)$ for all $u\in V\setminus U_{G\setminus E}(v,e)$, we have
\begin{align*}
&\sum_{u\in V\setminus \{v\}}\left(\frac{1}{d_{G\setminus (E\cup \{e\})}(u,v)} -\frac{1}{d_{G\setminus E}(u,v)}\right)\\
&=\sum_{u\in U_{G\setminus E}(v,e)}\left(\frac{1}{d_{G\setminus (E\cup \{e\})}(u,v)} -\frac{1}{d_{G\setminus E}(u,v)}\right).
\end{align*}
Applying a similar argument, we also have
\begin{align*}
&\obj(F\cup \{e\})-\obj(F)\\
&=\sum_{u\in U_{G\setminus F}(v,e)}\left(\frac{1}{d_{G\setminus (F\cup \{e\})}(u,v)} -\frac{1}{d_{G\setminus F}(u,v)}\right).
\end{align*}
Combining the above two equalities, we have
\begin{align*}
&\left(\obj(E\cup \{e\})-\obj(E)\right) - \left(\obj(F\cup \{e\})-\obj(F)\right)\\
&=\sum_{u\in U_{G\setminus E}(v,e)}\left(\frac{1}{d_{G\setminus (E\cup \{e\})}(u,v)} -\frac{1}{d_{G\setminus E}(u,v)}\right)\\
&\quad - \sum_{u\in U_{G\setminus F}(v,e)}\left(\frac{1}{d_{G\setminus (F\cup \{e\})}(u,v)} -\frac{1}{d_{G\setminus F}(u,v)}\right)\\
&\geq \sum_{u\in U_{G\setminus E}(v,e)}\left(\left(\frac{1}{d_{G\setminus (E\cup \{e\})}(u,v)} -\frac{1}{d_{G\setminus E}(u,v)}\right)\right.\\
&\qquad \qquad \qquad \ \ \left.-\left(\frac{1}{d_{G\setminus (F\cup \{e\})}(u,v)} -\frac{1}{d_{G\setminus F}(u,v)}\right)\right)\\
&= \sum_{u\in U_{G\setminus E}(v,e)}\left(\frac{1}{d_{G\setminus (E\cup \{e\})}(u,v)}-\frac{1}{d_{G\setminus (F\cup \{e\})}(u,v)}\right)\geq 0,
\end{align*}
where the first inequality follows from
$U_{G\setminus E}(v,e)\subseteq U_{G\setminus F}(v,e)$ and
$d_{G\setminus (F\cup \{e\})}(u,v)\geq d_{G\setminus F}(u,v)$ for any $u\in V$, 
and the second equality follows from $d_{G\setminus E}(u,v)=d_{G\setminus F}(u,v)$ for any $u\in U_{G\setminus E}(v,e)$.
Therefore, we have the theorem.
\end{proof}

We wish to remark that the above theorem can also be proved
using the fact that the objective function of the harmonic centrality maximization problem introduced in Crescenzi et al.~\cite{Crescenzi+16} is submodular.
However, if we employed such an approach, the notation would be more complex.
Therefore, we proved it from scratch, based on the definition of the submodularity.

\subsection{Pseudo-code of the greedy algorithm}\label{appendix:greedy}

See Algorithm~\ref{alg:greedy}.

\begin{algorithm}[t]
\caption{Greedy algorithm for Problem~\ref{prob:centrality_min}}
\label{alg:greedy}
\SetKwInOut{Input}{Input}
\SetKwInOut{Output}{Output} \Input{\ $G=(V,A)$, $v\in V$, and $b\in \mathbb{Z}_{>0}$}
\Output{\ $F\subseteq \rho(v)$ with $|F|\leq b$}
$F\leftarrow \emptyset$\;
\While{$|F|<b$}{
  Find $e\in \argmin \{\obj(e)\mid e\in \rho(v)\}$\;
  $F\leftarrow F\cup \{e\}$ and $G\leftarrow G\setminus \{e\}$\;
}
\Return $F$\;
\end{algorithm}

\subsection{Proof of Theorem~\ref{prop:greedy}}\label{appendix:proof_greedy}

\begin{proof}
We first show that any algorithm that outputs $F\subseteq \rho(v)$ with $|F|=b$ has an approximation ratio of $O(|V|)$.
Let $G=(V,A)$, $v\in V$, $b\in \mathbb{Z}_{>0}$ be an arbitrary instance of Problem~\ref{prob:centrality_min}.
If the optimal value is equal to $0$, then $b=|\rho(v)|$ holds, which means that the above algorithm can also achieve the optimal value.
If the optimal value is greater than $0$, then it is at least $1$, while any algorithm can achieve the objective value at most $h_G(v)\leq |V|-1$.
Therefore, the above algorithm has an approximation ratio of $O(|V|)$.

Next we show that the greedy algorithm has no approximation ratio of $o(|V|)$.
To this end, for any integer $k\geq 2$, we construct an instance of Problem~\ref{prob:centrality_min} as follows:
Let $G=(V,A)$ be a digraph.
The vertex set $V$ consists of two vertices $n_L, o_L$,
and size-$k$ subsets $N_R=\{n_R^1,\dots, n_R^k\}$ and $O_R=\{o_R^1,\dots, o_R^{k}\}$, in addition to one vertex $v'$.
The edge set $A$ is defined as the union of the following three sets:
\begin{align*}
&\{(n_L,v'), (n_R^i,v')\mid i=1,\dots, k\},\\
&\{(o_L,n_L)\},\\
&\{(o_R^i,n_R^j)\mid i = 1,\dots, k,\, j = 1,\dots, k\}.
\end{align*}
Then we employ $v'$ as target vertex $v$ and $k$ as budget $b$.

From now on, we analyze the performance of Algorithm~\ref{alg:greedy} for this instance.
In the first iteration, the removal of $(n_L,v)$ decreases the objective value by $3/2$,
while the removal of any incoming edge from $N_R$ decreases the objective value by $1$.
Therefore, the algorithm removes $(n_L,v)$.
In the later iterations, the removal of any edge decreases the objective value by $1$, and thus the algorithm removes $k-1$ edges from $N_R$ to $v$ and then terminates.
We see that the objective value of the output is $1+k/2$.
On the other hand, consider the following feasible (actually optimal) solution to the instance, i.e., all $k$ edges from $N_R$ to $v$.
Then the objective value of this solution is $1+1/2$.
Thus, the approximation ratio of the algorithm is lower bounded by
\begin{align*}
\frac{1+k/2}{1+1/2} = \Omega(|V|),
\end{align*}
meaning that Algorithm~\ref{alg:greedy} has no approximation ratio of $o(|V|)$ for Problem~\ref{prob:centrality_min}.
\end{proof}

\section{Omitted Contents from Section~\ref{sec:fast}}

\subsection{Proof of Theorem~\ref{thm:fast_limit}}\label{appendix:proof_tight}

\begin{proof}
For any integer $k\geq 2$, we construct an instance of Problem~\ref{prob:centrality_min} as follows:
Let $G=(V,A)$ be a digraph.
The vertex set $V$ consists of size-$k$ subsets $N_L=\{n_L^1,\dots, n_L^k\}$, $N_R=\{n_R^1,\dots, n_R^k\}$, $O_R=\{o_R^1,\dots, o_R^k\}$
and size-$(k(k-1))$ subset $O_L=\{o_L^1,\dots, o_L^{k(k-1)}\}$, in addition to one vertex $v'$.
The edge set $A$ is defined as the union of the following three sets:
\begin{align*}
&\{(n_L^i,v'), (n_R^i,v')\mid i=1,\dots, k\},\\ &\{(o_L^i,n_L^j)\mid i = 1,\dots, k(k-1),\, j=1,\dots, k,\, \lceil i/(k-1)\rceil =j\},\\
&\{(o_R^i,n_R^j)\mid i = 1,\dots, k, \, j = 1,\dots, k\}.
\end{align*}
Then we employ $v'$ as target vertex $v$ and $k$ as budget $b$.

From now on, we analyze the performance of Algorithm~\ref{alg:fast} for this instance.
For any $n_L^i\in N_L$ ($i=1,\dots, k$), we have $h_{G\setminus \rho(v)}(n_L^i)=k-1$.
For any $n_R^i\in N_R$ ($i=1,\dots, k$), we have $h_{G\setminus \rho(v)}(n_R^i)=k$.
Therefore, the output of Algorithm~\ref{alg:fast} is the set of all $k$ edges from $N_R$ to $v$,
which has an objective value of $k+k(k-1)/2$.
On the other hand, consider the following feasible (actually optimal) solution to the instance, i.e., the set of all $k$ edges from $N_L$ to $v$.
Then the objective value of this solution is $k+k/2$.
Thus, the approximation ratio of the algorithm is lower bounded by
\begin{align*}
\frac{k+k(k-1)/2}{k+k/2} = \frac{1+(k-1)/2}{1+1/2} = \Omega\left(\sqrt{|V|}\right),
\end{align*}
meaning that Algorithm~\ref{alg:fast} has no approximation ratio of $o\left(\sqrt{|V|}\right)$ for Problem~\ref{prob:centrality_min}.
Noticing that $h_G(v)\leq |V|-1$, we have the statement.
\end{proof}

\section{Omitted Contents from Section~\ref{sec:algorithm_cont}}

\subsection{Convergence result in Beck~\cite{Beck17}}\label{appendix:convergence}
For self-containedness, we review the convergence result of the projected subgradient method in Beck~\cite{Beck17},
on which the convergence result of Algorithm~\ref{alg:psm} is based.
Note that for simplicity, the description is appropriately specialized to our setting:
For example, we consider the objective function $f$ on $\mathbb{R}^n$ rather than a general vector space.

Let $f\colon \mathbb{R}^n\rightarrow (-\infty,\infty]$ and $C\subseteq \mathbb{R}^n$. Consider the following problem:
\begin{alignat*}{3}
&\text{(P)}\colon &\ &\text{minimize} &\quad &f(\bm{x})\\
&                 &  &\text{subject to} & &\bm{x}\in C.
\end{alignat*}

\begin{assumption}[Assumption 8.7 in Beck~\cite{Beck17}]\label{assump:1}
The following hold:
\begin{enumerate}
\item $f\colon \mathbb{R}^n \rightarrow (-\infty,\infty]$ is proper closed and convex;
\item $C\subseteq \mathbb{R}^n$ is nonempty closed and convex;
\item The interior of the (effective) domain of $f$ contains $C$;
\item The set of optimal solutions is nonempty, and the optimal value is denoted by $f^*$.
\end{enumerate}
\end{assumption}

We consider the projected subgradient method for (P), which is summarized in Algorithm~\ref{alg:psm_original}.
Note that we denote by $\proj_C(\bm{x})$ the projection of $\bm{x}\in \mathbb{R}^n$ onto $C$.
The sequence generated by the algorithm is $\{\bm{x}_t\}_{t\geq 0}$, while the sequence of function values generated by the algorithm is $\{f(\bm{x}^t)\}_{t\geq 0}$.
As the sequence of function values is not necessarily monotone, we are also interested in the sequence of best-achieved function values at or before $\ell$-th iteration,
which is defined as 
\begin{align*}
f^\ell_\text{best}=\min_{t=0,1,\dots, \ell} f(\bm{x}_t).
\end{align*}

If $C$ is compact, then there exists a constant $L_f>0$ for which $\|f'(\bm{x})\|\leq L_f$ for all $f'(\bm{x}) \in \partial f(\bm{x})$ and $\bm{x}\in C$,
where $\partial f(\bm{x})$ is the subdifferential of $f$ at $\bm{x}$.
Based on this, we have the following convergence result of Algorithm~\ref{alg:psm_original}:
\begin{fact}[A special case of Theorem 8.30 in Beck~\cite{Beck17}]\label{thm:psm_original}
Suppose that Assumption~\ref{assump:1} holds and assume that $C$ is compact.
Let $\Theta$ be an upper bound on the half-squared diameter of $C$, i.e.,
$\Theta\geq \max_{\bm{x},\bm{y}\in C} \frac{1}{2}\|\bm{x}-\bm{y}\|^2$.
Determine the stepsize $\eta_t$ ($t=0,1,\dots$) as $\eta_t=\frac{\sqrt{2\Theta}}{L_f\sqrt{t+1}}$.
Then for all $t\geq 2$, it holds that
\begin{align*}
f^t_\mathrm{best}-f^*\leq \frac{2(1+\log 3) L_f \sqrt{2\Theta}}{\sqrt{t+2}}.
\end{align*}
\end{fact}

\begin{algorithm}[t]
\caption{Projected subgradient method}
\label{alg:psm_original}
\SetKwInOut{Input}{Input}
\SetKwInOut{Output}{Output}
\Input{\ $\bm{x}_0\in C$ and some stopping condition}
\Output{\ $\bm{x}\in C$}
\While{the stopping condition is not satisfied}{
  Pick stepsize $\eta_t>0$ and subgradient $f'(\bm{x}_t)$ of $f'$ at $\bm{x}_t$\;
  $\bm{x}_{t+1}\leftarrow \proj_C\left(\bm{x}_{t} - \eta_t \cdot  f'(\bm{x}_{t})\right)$ and $t\leftarrow t+1$\;
}
\Return{$\bm{x}_t$}\;
\end{algorithm}

\begin{figure*}[h!]
\includegraphics[width=0.245\textwidth]{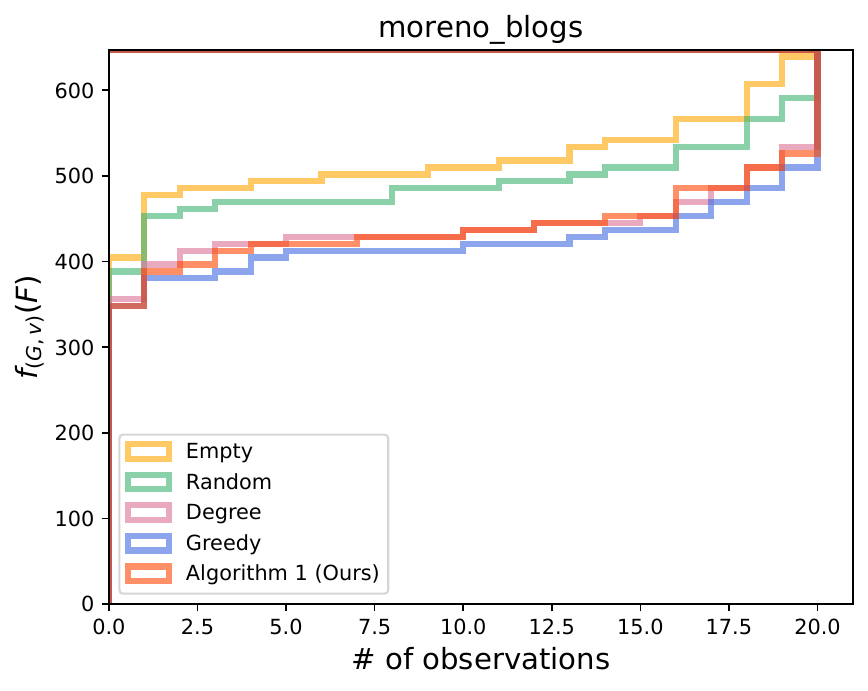}
\includegraphics[width=0.245\textwidth]{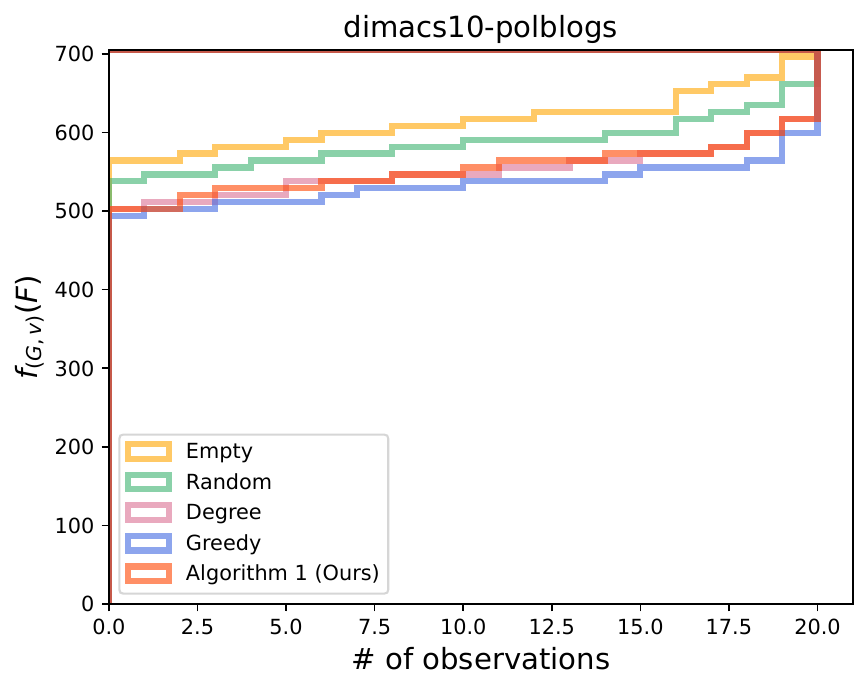}
\includegraphics[width=0.245\textwidth]{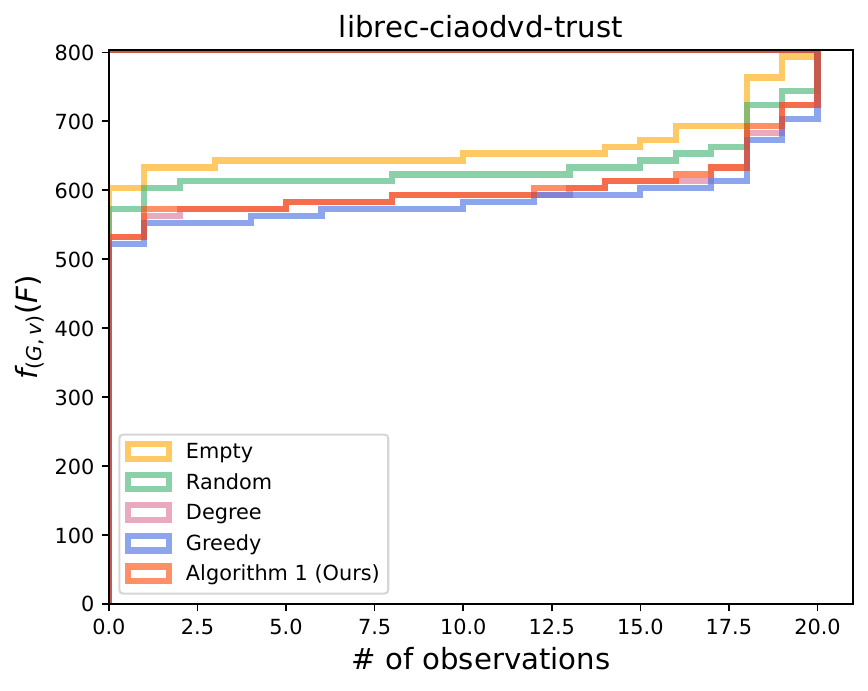}
\includegraphics[width=0.245\textwidth]{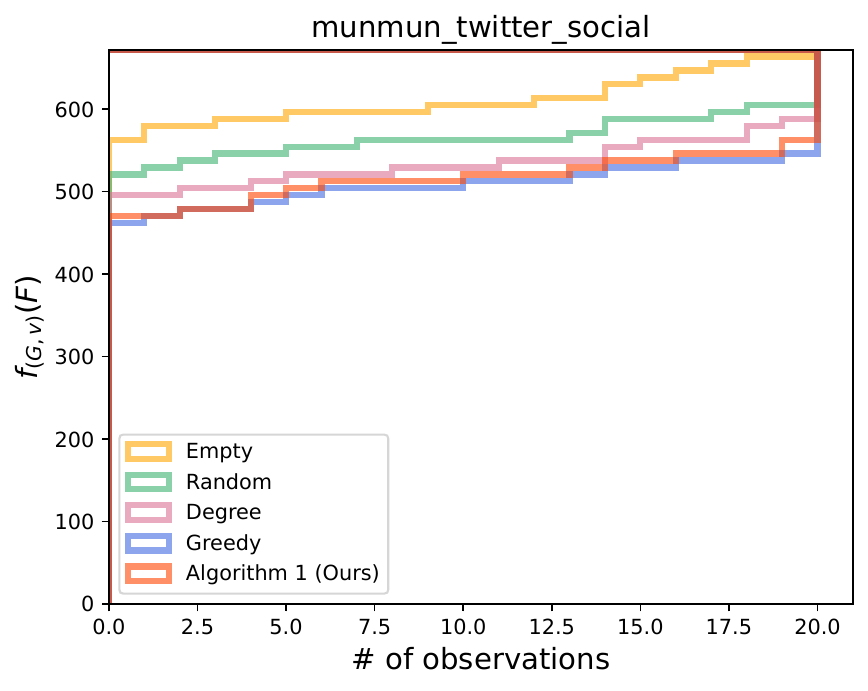}
\includegraphics[width=0.245\textwidth]{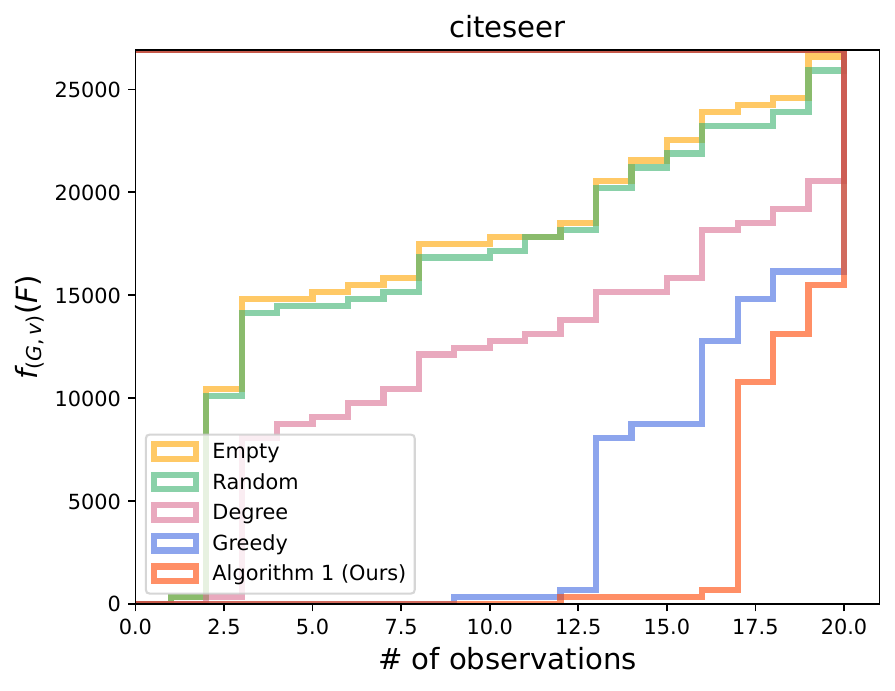}
\includegraphics[width=0.245\textwidth]{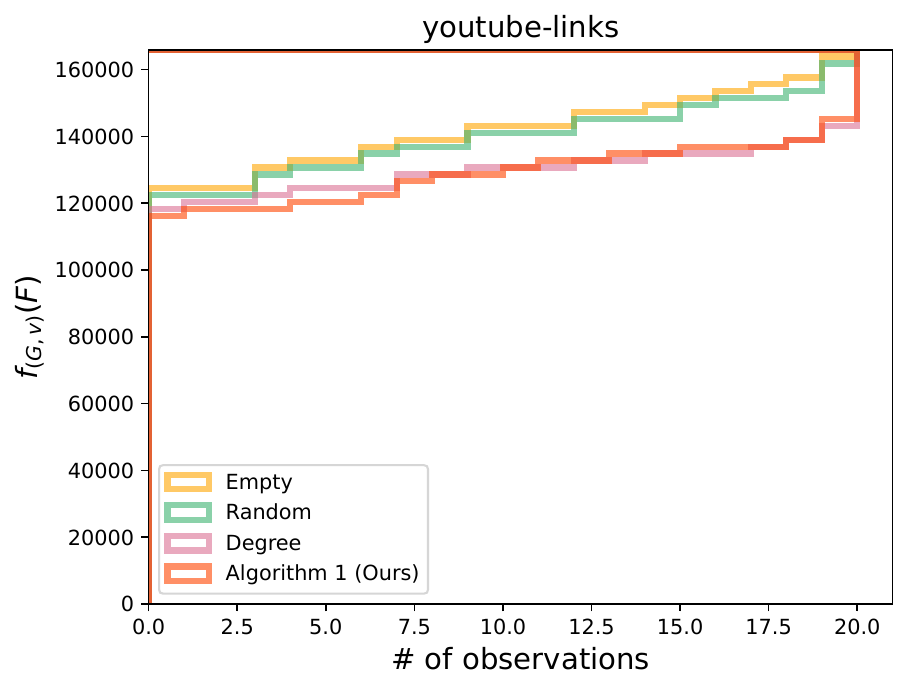}
\includegraphics[width=0.245\textwidth]{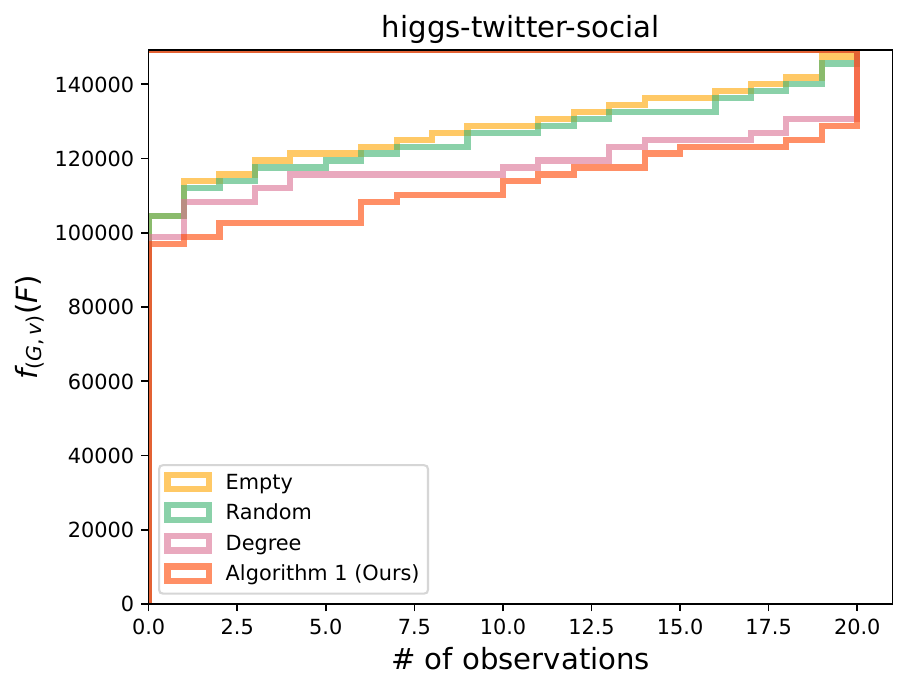}
\includegraphics[width=0.245\textwidth]{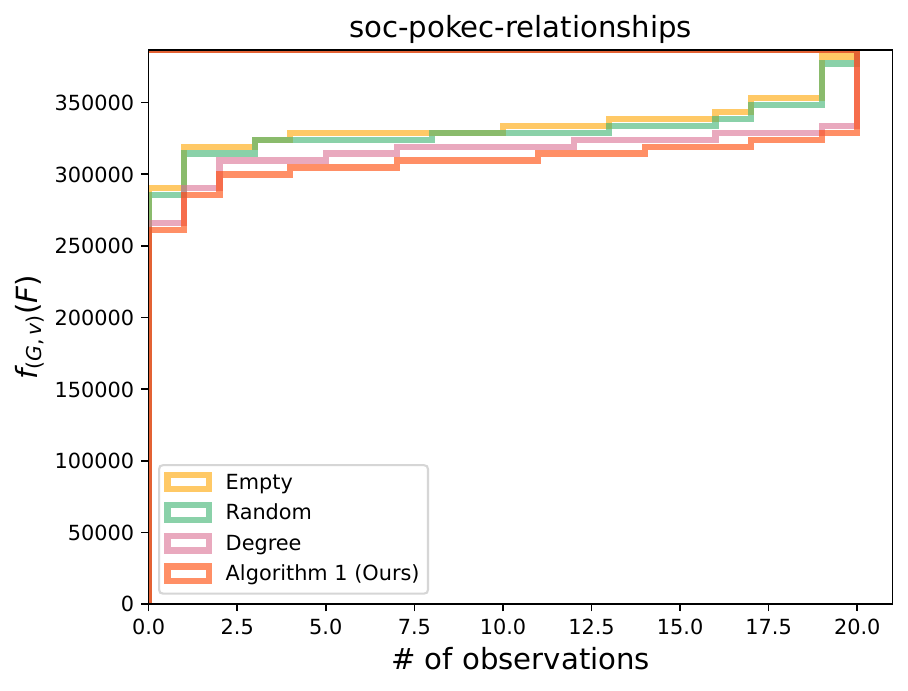}
\caption{Quality of solutions of the algorithms (except for Algorithm~\ref{alg:biapprox}) with $b=\lfloor \frac{1}{4}|\rho(v)|\rfloor$.}\label{fig:performance_small}
\end{figure*}

\begin{figure*}
\includegraphics[width=0.245\textwidth]{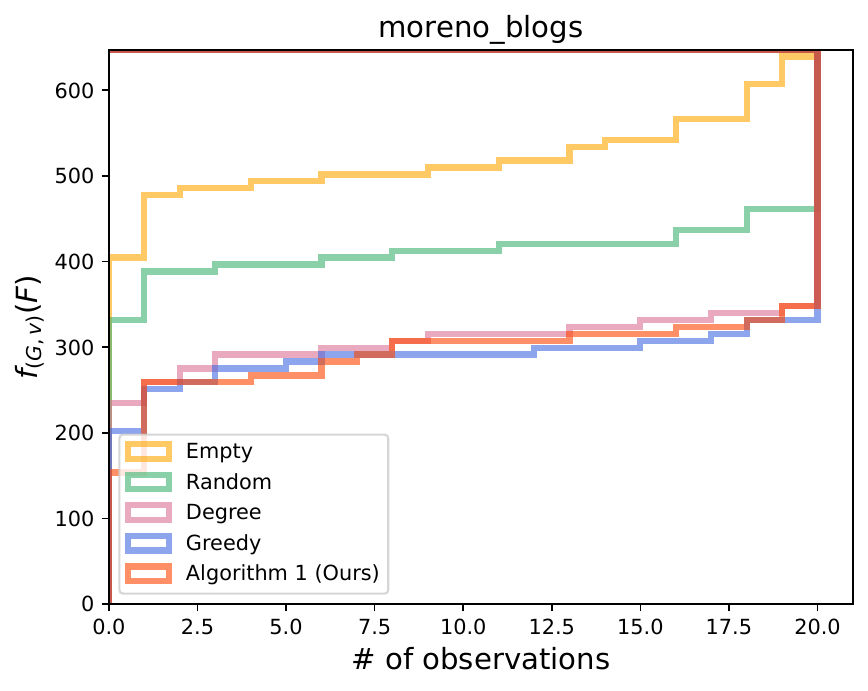}
\includegraphics[width=0.245\textwidth]{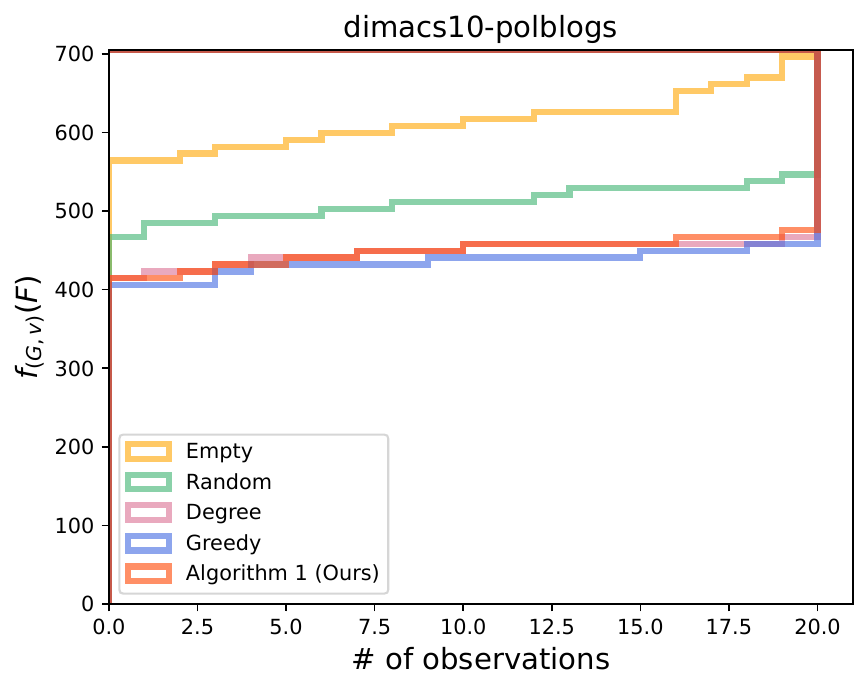}
\includegraphics[width=0.245\textwidth]{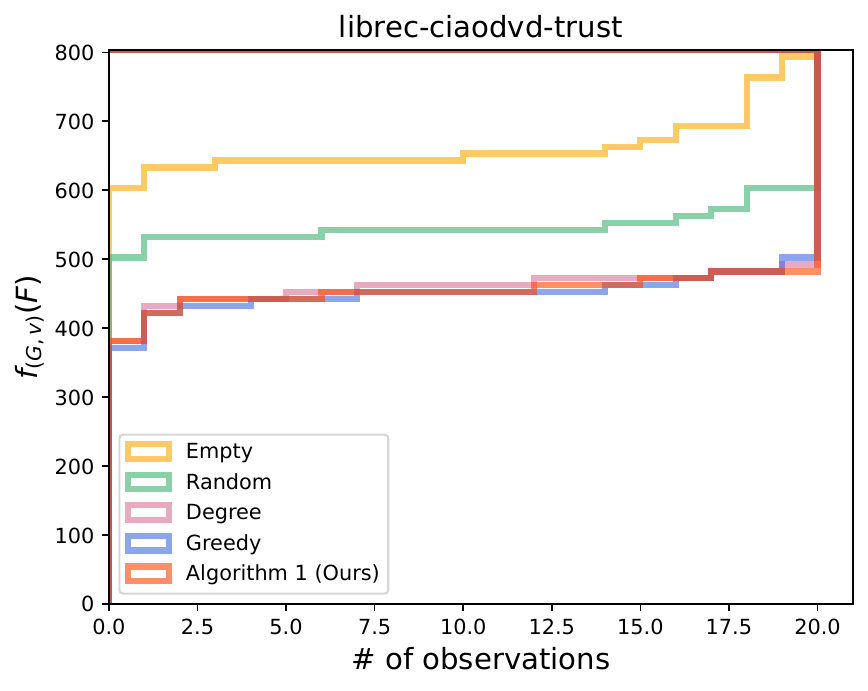}
\includegraphics[width=0.245\textwidth]{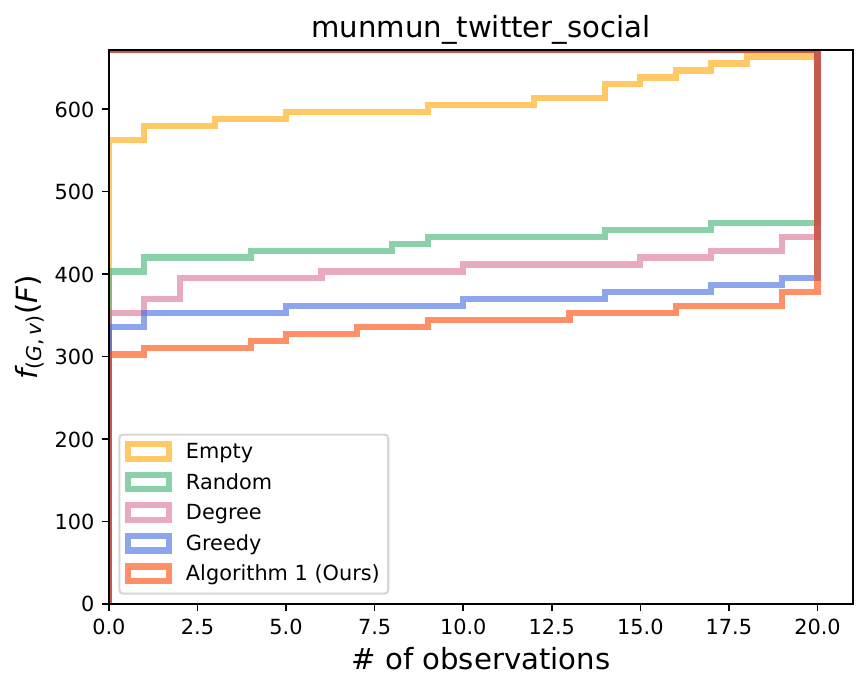}
\includegraphics[width=0.245\textwidth]{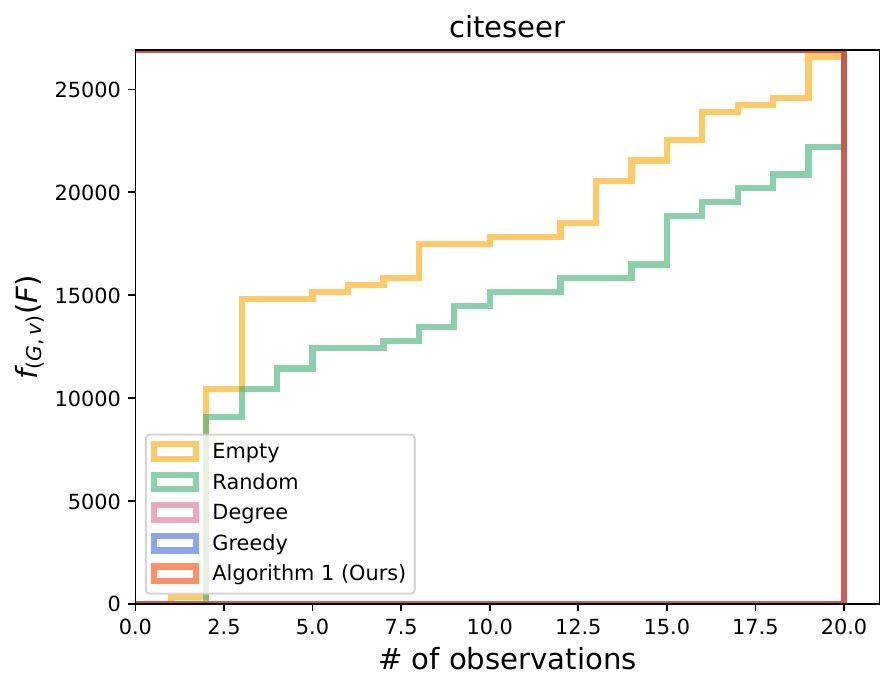}
\includegraphics[width=0.245\textwidth]{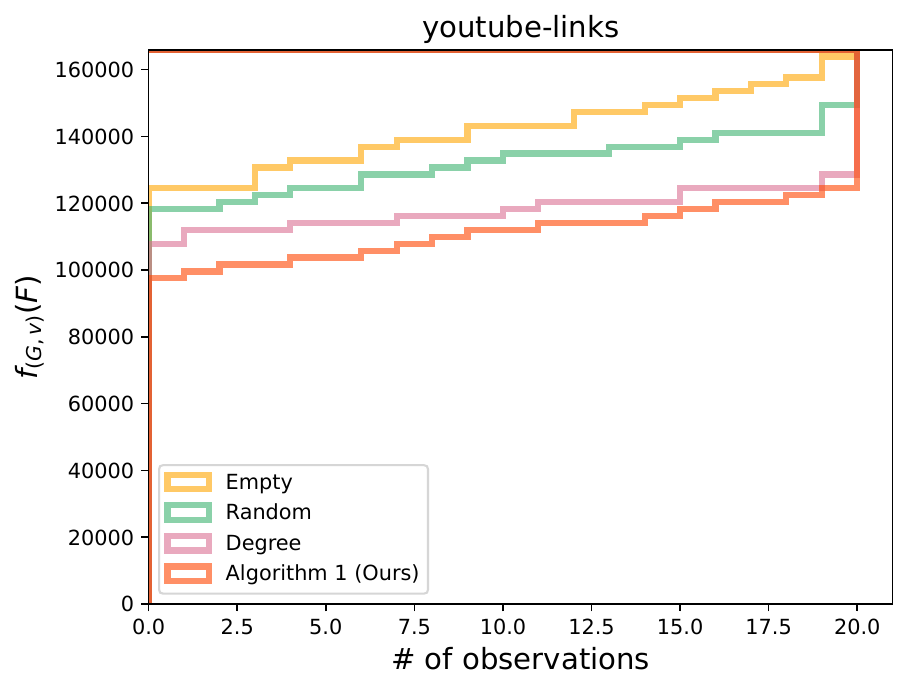}
\includegraphics[width=0.245\textwidth]{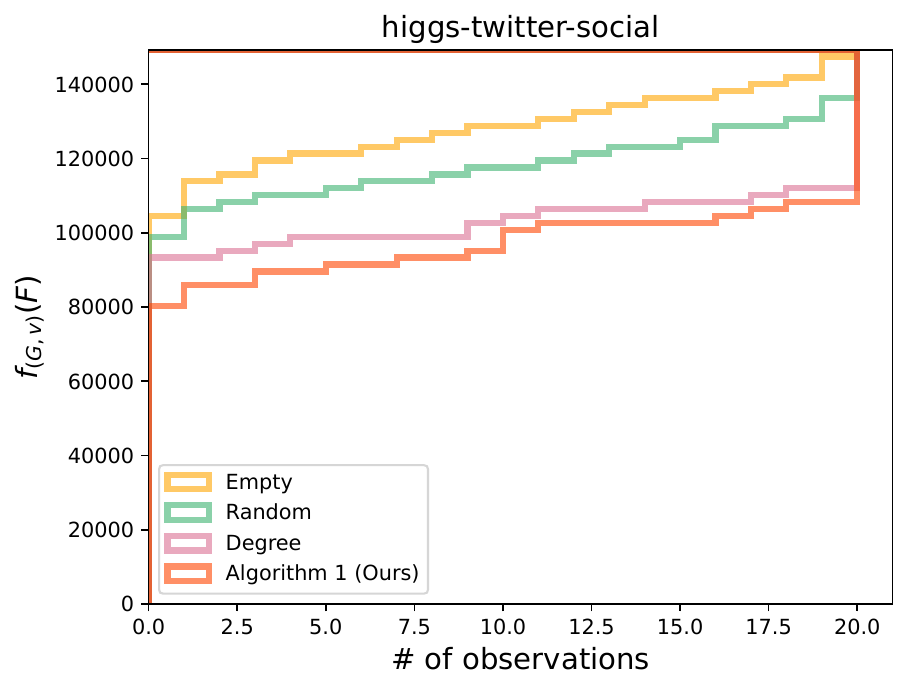}
\includegraphics[width=0.245\textwidth]{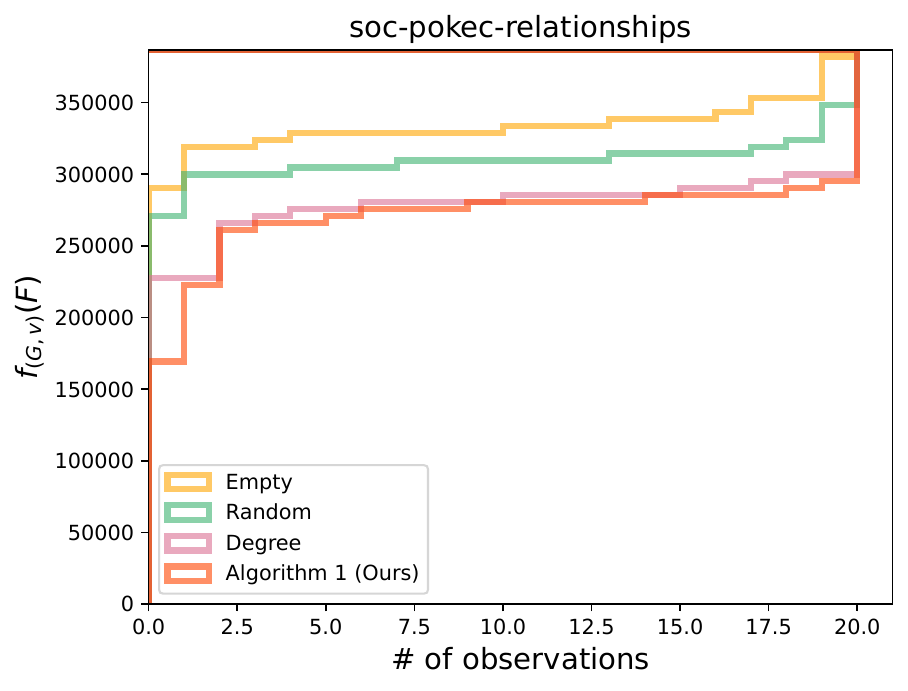}
\caption{Quality of solutions of the algorithms (except for Algorithm~\ref{alg:biapprox}) with $b=\lfloor \frac{3}{4}|\rho(v)|\rfloor$.}\label{fig:performance_large}
\end{figure*}

\subsection{Proof of Theorem~7}\label{appendix:proof_convergence}
\begin{proof}
The proof is almost straightforward from Fact~\ref{thm:psm_original} in Appendix~\ref{appendix:convergence}.
By defining $\widehat{f}_{(G,v)}(\bm{x})=\infty$ for any $\bm{x}\in \mathbb{R}^{\rho(v)} \setminus C$,
we have an extended real-valued function $\widehat{f}_{(G,v)}\colon \mathbb{R}^{\rho(v)}\rightarrow (-\infty,\infty]$.
Then \textsf{Relaxation} becomes a special case of (P) in Appendix~\ref{appendix:convergence}.
Let us confirm the validity of Assumption~\ref{assump:1} in Appendix~\ref{appendix:convergence} for \textsf{Relaxation}:
\begin{itemize}
\item[(1)] From the definition of $\widehat{f}_{(G,v)}$ (and $f_{(G,v)}$),
$\widehat{f}_{(G,v)}$ does not attain $-\infty$ and $\widehat{f}(\bm{x})<\infty$ for any $\bm{x}\in C\neq \emptyset$,
which means that $\widehat{f}_{(G,v)}$ is proper.
As $\widehat{f}_{(G,v)}$ is continuous over $C$ and the (effective) domain of $\widehat{f}_{(G,v)}$ (i.e., $C$) is closed,
by Theorem~2.8 in Beck~\cite{Beck17}, $\widehat{f}_{(G,v)}$ is closed.
As mentioned above, $\widehat{f}_{(G,v)}$ is convex by the submodularity of $f_{(G,v)}$.
\item[(2)] This is trivial.
\item[(3)] As the boundary of $C$ is not contained in the interior of the (effective) domain of $\widehat{f}_{(G,v)}$, this does not hold.
However, according to the analysis in Beck~\cite{Beck17}, this assumption is used only for guaranteeing the subdifferentiability of the objective function over the feasible region (i.e., the subdifferentiability of $\widehat{f}_{(G,v)}$ over $C$ in our case),
which is clearly valid.
\smallskip
\item[(4)] From the boundedness of $C$ and the piecewise-linearity of $\widehat{f}_{(G,v)}$,
we see that the set of optimal solutions to \textsf{Relaxation} is nonempty. 
\end{itemize}
Finally, $C$ is obviously compact.
Therefore, we have the theorem.
\end{proof}

\section{Omitted Contents from Section~\ref{sec:experiments}}

\subsection{Counterparts of Figure~\ref{fig:performance}}\label{appendix:figures}
See Figures~\ref{fig:performance_small} and~\ref{fig:performance_large}.
Notably, for \texttt{citeseer} with $b=\lfloor \frac{1}{4}|\rho(v)|\rfloor$ and \texttt{munmun-twitter-social} with $b=\lfloor \frac{3}{4}|\rho(v)|\rfloor$, Algorithm~\ref{alg:fast} performs much better than \textsf{Greedy}, demonstrating that \textsf{Greedy} does not necessarily perform better than Algorithm~\ref{alg:fast} even for real-world graphs small enough to allow for running \textsf{Greedy}.

\end{document}